\documentclass[dvips,aos,preprint]{imsart}

\setattribute{journal}{name}{}

\RequirePackage[OT1]{fontenc}
\RequirePackage{amsthm,amsmath,amssymb}
\RequirePackage[colorlinks]{hyperref}
\RequirePackage{}
\usepackage{graphicx}

\startlocaldefs
\numberwithin{equation}{section}
\theoremstyle{plain}
\newtheorem{thm}{Theorem}
\newtheorem{lemma}{Lemma}
\newtheorem{corollary}{Corollary}
\newtheorem{proposition}{Proposition}
\theoremstyle{definition}

\newtheorem{proc}{Algorithm}
\theoremstyle{remark}
\newtheorem{remark}{Remark}

\newtheorem{assumption}{Assumption}
\endlocaldefs

\newcommand{\convP}{\overset{p}{\longrightarrow}}
\newcommand{\convL}{\overset{L}{\longrightarrow}}
\newcommand{\bz}{{\bf z}}
\newcommand{\bt}{{\bf t}}
\newcommand{\bu}{{\bf u}}
\newcommand{\bv}{{\bf v}}
\newcommand{\bw}{{\bf w}}
\newcommand{\bdr}{{\bf r}}
\newcommand{\bs}{{\bf s}}
\newcommand{\bx}{{\bf x}}
\newcommand{\bA}{{\bf A}}
\newcommand{\bB}{{\bf B}}
\newcommand{\bD}{{\bf D}}
\newcommand{\bS}{{\bf S}}
\newcommand{\bU}{{\bf U}}
\newcommand{\bV}{{\bf V}}
\newcommand{\bX}{{\bf X}}
\newcommand{\bZ}{{\bf Z}}
\newcommand{\btheta}{\mbox{\boldmath$\theta$}}
\newcommand{\bTheta}{\mbox{\boldmath$\Theta$}}
\newcommand{\bbeta}{\mbox{\boldmath$\beta$}}

\newcommand{\bomega}{\mbox{\boldmath$\omega$}}
\newcommand{\bOmega}{\mbox{\boldmath$\Omega$}}
\newcommand{\bSigma}{\mbox{\boldmath$\Sigma$}}

\newcommand{\bRe}{{\mathbb R}}

\begin{document}

\begin{frontmatter}
\title{Statistical Inference on Transformation Models:
a Self-induced Smoothing Approach}
\runtitle{Self-induced Smoothing for Transformation Models}

\begin{aug}
\author{\fnms{Junyi} \snm{Zhang}\ead[label=e1]{jz2299@columbia.edu}},
\author{\fnms{Zhezhen} \snm{Jin}\ead[label=e2]{zjin@biostat.columbia.edu}},\\
\author{\fnms{Yongzhao} \snm{Shao}\ead[label=e3]{yongzhao.shao@nyumc.org}}
\and
\author{\fnms{Zhiliang} \snm{Ying}\ead[label=e4]{zying@stat.columbia.edu}}

\runauthor{J. ZHANG, Z. JIN, Y. SHAO AND Z. YING}

\affiliation{Columbia University and New York University}

\address{
J. Zhang\\
Z. Ying\\
Department of Statistics,\\
Columbia University,\\
2690 Broadway,\\
New York, NY10027, U.S.A.\\
\printead{e1}\\
\printead{e4}}

\address{
Z. Jin\\
Department of Biostatistics,\\
Mailman School of Public Health,\\
Columbia University,\\
New York, NY10032, U.S.A.\\
\printead{e2}}

\address{
Y. Shao\\
Division of Biostatistics,\\
School of Medicine,\\
New York University,\\
New York, NY10016, U.S.A.\\
\printead{e3}}
\end{aug}

\begin{abstract}
This paper deals with a general class of transformation models that contains many important
semiparametric regression models as special cases. It
develops a self-induced smoothing for the maximum rank correlation estimator, resulting
in simultaneous point and variance estimation. The self-induced smoothing does not require
bandwidth selection, yet provides the right amount of smoothness so that the estimator
is asymptotically normal with mean zero (unbiased) and variance-covariance matrix consistently
estimated by the usual sandwich-type estimator. An iterative algorithm is given for the
variance estimation and shown to numerically converge to a consistent limiting variance
estimator. The approach is applied to a data set involving survival times of primary biliary
cirrhosis patients. Simulations results are reported, showing that
the new method performs well under a variety of scenarios.
\end{abstract}


\end{frontmatter}

\section{Introduction}
Consider the following class of regression models, with response variable denoted by $Y$
and $(d+1)$-dimensional covariate vector by $\bX$,
\begin{equation}\label{eqn: model}
Y=H(\bX'\bbeta + \varepsilon)
\end{equation}
where $\bbeta$ is the unknown
parameter vector, $\varepsilon$ is the unobserved error term that is
independent of $\bX$ with a completely unspecified distribution,
and $H$ is a monotone increasing, but otherwise unspecified function.

It is easily seen that this class of models contains many commonly used regression models as its submodels that are especially important in the econometrics and survival analysis literature.
For example,
with $H(u)=u$, \eqref{eqn: model} becomes the standard regression model with an unspecified error distribution;
with $H(u)=u^{\lambda}$ ($\lambda>0$), the Box-Cox transformation model (Box and Cox, 1964);
with $H(u)=I[u\geq 0]$, the binary choice model (Maddala, 1983; McFadden, 1984); with  $H(u)=uI[u\geq 0]$,
a censored regression model (Tobin, 1958; Powell, 1984);
with $H(u)=\exp (u)$, the accelerated failure times
(AFT) model (Cox and Oakes, 1984;
Kalbfleisch and Prentice, 2002);
with $\varepsilon$ having an extreme
value density $f(w)=\exp(w-\exp(w))$, the Cox
proportional hazards regression (Cox, 1972);
with $\varepsilon$ having the standard logistic distribution, the proportional
odds regression (Bennett, 1983).

A basic tool for handling model \eqref{eqn: model} is
the maximum rank correlation (MRC) estimator proposed in the econometrics
literature by Han (1987).
Because both the transformation function $H$ and the error distribution are unspecified,
not all components of $\bbeta$ are identifiable. Without loss of generality,
we shall assume henceforth that the last component, $\beta_{d+1}=1$.
Let $(Y_1, \bX_1), ..., (Y_n, \bX_n)$ be a random sample from \eqref{eqn: model}.
Han's MRC estimator, denoted by $\hat{\btheta}_n$,
 is the maximizer of following objective function
\begin{equation}\label{eqn: RC}
Q_n(\btheta)={1 \over n(n-1)}\sum_{i\not=j}
I[{Y_i>Y_j]}I[{\bX_i'\bbeta(\btheta)>\bX_j'\bbeta(\btheta)}],
\end{equation}
where $I[\ \cdot\ ]$ denotes the indicator function, $\bX'$ the transpose of $\bX$, and $\btheta$
the first $d$ components of $\bbeta$, i.e. $\bbeta(\btheta)=(\theta_1,...,\theta_d,1)'$.
Han (1987) proved that the MRC estimator $\hat{\btheta}_n$ is strongly
consistent under certain regularity conditions.

An important subsequent development is due to Sherman (1993), who made use of the empirical process theory and Hoeffding's decomposition
to approximate the objective function, viewed as a U-process.
He showed that $\hat{\btheta}_n$ is, in fact, asymptotically normal under
additional regularity conditions. Estimation of the transformation function $H$ was studied by Chen (2002), who constructed a rank-based estimator and established its consistency and asymptotic normality.

In addition to the econometrics, model \eqref{eqn: model} also encompasses the main semiparametric models
in survival analysis, where right censoring is a major feature. Under the right censorship,
there is a censoring variable $C$ and one observes $\tilde Y=Y \wedge C$ and
 $\Delta_i=I(Y_i \leq C_i)$.
Khan and Tamer (2007) constructed the following partial rank correlation function
as an extension of the rank correlation objective function \eqref{eqn: RC},
\begin{equation} \label{eqn: PRC}
Q_n^*(\btheta)={1 \over n(n-1)}\sum_{i\not=j}
\Delta_j I[{\tilde Y_i>\tilde Y_j]}I[{\bX_i'\bbeta(\btheta)>\bX_j'\bbeta(\btheta)}].
\end{equation}
 They showed that the resulting maximum partial rank correlation estimate (PRCE) $\hat{\btheta}_n^*$,
 as the maximizer of $Q_n^*(\btheta)$, is consistent and asymptotically normal.

Crucial for the statistical inference of \eqref{eqn: model}
based on $\hat\btheta_n$ is the consistent variance estimation.
In standard objective (loss) function derived estimation,
the asymptotic variance is usually
estimated by a sandwich-type estimator of form $\hat \bA^{-1}\hat \bV \hat \bA^{-1}$
with $\hat \bA$ being the second derivative of the objective function
and $\hat \bV$ an estimator of the variance of the first derivative (score).
The challenge here, however, is that
$Q_n$ itself is a (discontinuous) step function that
precludes automatic use of differentiation to obtain $\hat \bA$.
Furthermore, $\hat \bV$ is also difficult to obtain since
the score function cannot be derived directly from $Q_n$ via differentiation.
Sherman(1993) suggested using numerical derivatives of first and second orders
to construct $\hat \bA$ and $\hat \bV$.
His approach requires bandwidth selection for the derivative functions.
It is unclear how stable the resulting variance estimator is. Alternatively, one may resort to
bootstrap (Efron, 1979) or other resampling methods (e.g. Jin et al., 2001). These approaches require
repeatedly solving the maximization of \eqref{eqn: RC}, which is discontinuous and often multidimensional when
$d>1$. The computational cost could therefore be prohibitive.

In this paper, we develop a self-induced smoothing method
for rank correlation criterion function \eqref{eqn: RC} so that the differentiation
can be performed, while bypassing the bandwidth selection.
Both point and variance estimators can be obtained
simultaneously in a straightforward way that is typically used for smooth objective functions.
The new method is motivated by a novel approach proposed in Brown and Wang (2005, 2007),
where an elegant self-induced smoothing method was introduced for non-smooth estimating functions. Although our approach bears similarity with that of Brown and Wang (2005), it is far from clear why such self-induced smoothing is suitable for the discrete objective function (rank correlation).  In fact, undersmoothing would make the Hessian (second derivative) unstable while oversmoothing would introduce significant bias. Through highly technical and tedious derivations, we will show that the proposed method does strike a right balance in terms of asymptotic unbiasedness and enough smoothness for differentiation (twice).

The rest of the paper is organized as follows.
In Section 2, the new methods are described and related large sample properties are developed.
In particular, we give construction for simultaneous point and variance estimation and
show that the resulting point estimator is asymptotically normal and
the variance estimator is consistent.
In Section 3, the approach, along with the algorithm and large sample properties, is extended to handle survival data with right censoring.
Simulation results are reported in Section 4, where application to
a real data set is also given. Section 5 contains some concluding remarks.
Additional technical proofs can be found in the Appendix.

\section{Main Results}

In this section we develop a self-induced smoothing method
for the rank correlation criterion function defined by \eqref{eqn: RC}.
It is divided into three subsections, with the first introducing the method and the algorithm, the second establishing large sample properties and the third covering proofs.

\subsection{Methods}

Since MRC estimator $\hat\btheta_n$ is asymptotically normal (Sherman, 1993),
its difference with the true parameter value,
$\hat \btheta_n - \btheta$,
should approximately be a Gaussian noise $\bZ/\sqrt{n}$, where
 $\bZ \sim N({\bf 0}, \bSigma)$ is a $d$-dimensional normal random vector
with mean $\bf 0$ and covariance matrix $\bSigma$.
Assume that $\bZ$ is independent of data and let $E_{\bZ}$ denote the expectation with respect to $\bZ$ given data.
A self-induced smoothing for $Q_n$ is $\tilde{Q}_n(\btheta)=E_{\bZ} Q_n(\btheta+\bZ/\sqrt{n})$.
The self-induced smoothing using the limiting Gaussian distribution was originally proposed by
Brown and Wang (2005) for certain non-smooth estimating functions.

To get an explicit
form for $\tilde{Q}_n$, let $\Phi$ be the standard normal distribution function,
$\bX_{ij}=\bX_i-\bX_j$, $\sigma_{ij}=\sqrt{(\bX_{ij}^{(1)})'\bSigma \bX_{ij}^{(1)}}$ where
$\bX_{ij}^{(1)}$ denotes the first $d$ components of $\bX_{ij}$. Then, it is easy to see that
\setcounter{equation}{3}
\begin{equation} \label{eqn: SRC}
\tilde{Q}_n(\btheta)=\frac{1}{n(n-1)}\sum_{i\neq j}
I[Y_i>Y_j]\Phi\left(\sqrt{n}\bX_{ij}'\bbeta(\btheta)/\sigma_{ij}\right).
\end{equation}
We shall use $\tilde{\btheta}_n = \arg\max_{\bTheta} \tilde{Q}_n(\btheta)$ to denote the corresponding estimator, which will be called the smoothed maximum rank correlation estimator (SMRCE). Here and in the sequel, $\bTheta$ denotes the parameter space for $\btheta$.

\begin{remark}
Smoothing is an appealing way for a simple solution to the inference problem
associated with the MRCE. If $\tilde{Q}_n$ were a usual smooth objective
function, then its first derivative would become the score function and its second derivative
could be used for variance estimation. Speficically, if we use $\bV$ to denote the limiting variance
of the score scaled by $n$ and $\bA$ the limit of the second derivative, then the asymptotic variance of the resulting estimator, scaled by $ n$, should be of form $\bA^{-1}\bV\bA^{-1}$. A consistent
estimator could then be obtained by the plug-in method, i.e. replacing unknown parameters by their corresponding empirical estimators.
\end{remark}
\begin{remark}
It is unclear, however, whether or not the self-induced smooth will provide
a right amount of smoothing, even in view of the results given in Brown and Wang (2005).
With over-smoothing, $\tilde\btheta_n$ may be asymptotically biased, i.e. the bias is not of
order $o(n^{-1/2})$; with under-smoothing, the ``score'' function (first derivative of
$\tilde{Q}_n$) may have multiple ``spikes'' and thus the second derivative matrix (Hessian) of
$\tilde{Q}_n$ may not behave properly and certainly cannot be expected to provide a consistent
variance estimator.
\end{remark}

In Subsection 2.2, we show that the self-induced smoothing here does result in a right amount of
smoothing in the sense that the bias is asymptotically negligible and the Hessian matrix behave
properly. Before starting the theoretic developments, we first describe our method.

We first differentiate the smoothed objective function $\tilde{Q}_n$ to get score
\[\tilde \bS_n(\btheta)=\frac{1}{n(n-1)}\sum_{i<j}
H_{ij}\phi\left(\frac{\sqrt{n}\bX_{ij}'\bbeta(\btheta)}{\sigma_{ij}}\right)
\frac{\sqrt{n}\bX_{ij}^{(1)}}{\sigma_{ij}},\]
where $H_{ij}=sgn(Y_i-Y_j)$.
This is a U-process of order 2 with kernel
\begin{equation*}\bs_n(\bU_i,\bU_j)=\frac{1}{2}
H_{ij}\phi\left(\frac{\sqrt{n}\bX_{ij}'\bbeta(\btheta)}{\sigma_{ij}}\right)
\frac{\sqrt{n}\bX_{ij}^{(1)}}{\sigma_{ij}},\end{equation*}
where $\bU_i$ denotes the pair $(Y_i,\bX_i)$.

By Hoeffding's decomposition,
the asymptotic variance of $\sqrt n\tilde \bS_n(\btheta)$ is approximated by
\begin{eqnarray}\label{eqn: Vn}
\begin{split}
\hat{\bV}_n(\btheta, \bSigma)
&=\frac{1}{n^3}\sum_{i=1}^n\left\{\sum_{j}
\left[H_{ij}\times\phi\left(\frac{\sqrt{n}\bX_{ij}'\bbeta}{\sigma_{ij}}\right)
\frac{\sqrt{n}\bX_{ij}^{(1)}}{\sigma_{ij}}\right]\right\}^{\otimes2},\\
\end{split}
\end{eqnarray}
where, for a vector $\bv$,  $\bv^{\otimes2}=\bv\bv'$.
Thus, $\hat{\bV}_n(\hat\btheta_n, \bSigma)$ is used to estimate $\bV$, the middle part
of the ``sandwich'' variance formula discussed in Remark 1.

As for $\bA$, we differentiate $\tilde \bS_n(\btheta)$ to get
\begin{equation} \label{eqn: An}
\hat{\bA}_n(\btheta,\bSigma)
=\frac{1}{2n(n-1)}\sum_{i \neq j}\left\{ H_{ij}
\times\dot{\phi}\left(\frac{\sqrt{n}\bX_{ij}'\bbeta}{\sigma_{ij}}\right)
\left[\frac{\sqrt{n}\bX_{ij}^{(1)}}{\sigma_{ij}}\right]^{\otimes2}\right\},
\end{equation}
where $\dot{\phi}(z)=-z\phi(z)$ is the derivative of $\phi(z)$. Although the self-induced smoothing was motivated
earlier with $\bSigma$ being the limiting covariance matrix of the estimator, we will show later that for
any positive definite matrix $\bSigma$,
$\hat{\bA}_n(\hat\btheta_n, \bSigma)$ converges to $\bA$.

Note that the above discussions about $\bA$ and $\bV$ are not mathematically rigorous. This is because
the kernel function for the score process is sample size $n$-dependent. The usual asymptotic theory
for the U-process is not applicable. Indeed, our rigorous derivations, to be given in Subsection 2.3,
are quite tedious, involving many approximations that are quite delicate.

\smallskip
Let
\begin{equation}\label{eqn: var formula}
\hat{\bD}_n(\btheta, \bSigma)=
\hat{\bA}_n^{-1}(\btheta, \bSigma)\times
\hat{\bV}_n(\btheta, \bSigma)\times
\hat{\bA}_n^{-1}(\btheta, \bSigma).
\end{equation}
If $\btheta$ is the true parameter value, then $\hat{\bD}_n(\btheta, \bSigma)$ converges to the limiting covariance
matrix, which is the desired choice for $\bSigma$ in the self-induced smoothing. Therefore, \eqref{eqn: var formula}
leads to an iterative algorithm of form $\hat{\bSigma}_n^{(k)}=\hat\bD_n(\hat{\btheta}_n,\hat{\bSigma}_n^{(k-1)})$;
see also Brown and Wang (2005). Specifically, we propose the following algorithm:

\begin{proc} (SMRCE)
\begin{enumerate}{\leftmargin=-1em}
\item Compute the MRC estimator $\hat{\btheta}_n$ and set $\hat\bSigma^{(0)}$ to be the identity matrix.

\item Update variance-covariance matrix $\hat{\bSigma}_n^{(k)}=\hat\bD_n(\hat{\btheta}_n,\hat{\bSigma}_n^{(k-1)})$.
Smooth the rank correlation $Q_n(\btheta)$ using covariance matrix $\hat{\bSigma}_n^{(k)}$.
Maximize the resulting smoothed rank correlation to get an estimator $\hat{\btheta}_n^{(k)}$.

\item Repeat step 2 until $\hat{\btheta}_n^{(k)}$ converge.
\end{enumerate}
\end{proc}

\subsection{Large-sample properties}
This subsection is devoted to the large sample theory. The main results are: 1. the smoothed MRC estimator (SMRCE)
is asymptotically equivalent to the MRC estimator; 2.
the proposed method leads to a consistent variance estimator; and 3.
the iterative algorithm for point and variance estimation converges numerically.

We first introduce notation as well as assumptions, which are similar to those in Sherman (1993) for
the MRC estimator. Let
\begin{equation}\label{eqn: tau}
\tau(y,\bx,\btheta)=E \left[I_{[y>Y]}I_{[(\bx-\bX)'\bbeta(\btheta)>0]}+
I_{[y<Y]}I_{[(\bx-\bX)'\bbeta(\btheta)<0]}\right], \end{equation}
which  is the projection of the kernel of U-process $Q_n(\btheta)$. The
expectation is taken for $(\bX, Y)$. Also let
\[|\nabla_m |\tau(y, \bx, \btheta)=
\sum_{i_1, ..., i_m}\left| {\partial^m \tau(y, \bx, \btheta)\over
\partial \theta_{i_1}\cdots \partial\theta_{i_m}}\right|. \]

\smallskip
The following Assumptions 1 and 2 are used in Han (1987) (see also Sherman, 1993) to
establish consistency for the MRC estimator. For asymptotic normality, we need
an additional regularity condition (Assumption 3) given in Sherman (1993).
\begin{assumption}
The true parameter value $\btheta_0$ 
is an interior point of $\bTheta$, which is a compact subset of the $d$-dimensional Euclidean space $\bRe^d$.
\end{assumption}
\begin{assumption}
The support of $\bX$ is not contained in any linear subspace of $\bRe^{d+1}$.
Conditional on the first $d$ components of $\bX$, the last component of $\bX$ has a density
function with respect to the Lebesgue measure.
\end{assumption}
\begin{assumption}
There exists a neighborhood, $\cal N$, of $\btheta_0$ such that
for each pair $(y, \bx)$ of possible values of $(Y, \bX)$,

(i) The second derivatives of $\tau(y, \bx; \btheta)$ with respect to $\btheta$
exist in ${\cal N}$.

(ii) There is an integrable function $M_1(y, \bx)$ such that for all $\btheta$ in
${\cal N}$, \[\|\nabla_2\tau(y,\bx; \btheta)-\nabla_2\tau(y,\bx; \btheta_0)\|_2\leq M_1(y,\bx)|\btheta-\btheta_0|.\]

(iii) $E(|\nabla_1|\tau(Y, \bX; \btheta_0))^2 <+\infty.$

(iv) $E|\nabla_2|\tau(Y, \bX; \btheta_0)<+\infty.$

(v) The matrix $E\nabla_2\tau(Y, \bX; \btheta_0)$ is strictly negative definite.
\end{assumption}

\begin{proposition} (Sherman, 1993)
Assume that Assumptions 1-3 hold. We have,
uniformly over any $o_p(1)$ 
neighborhood of $\btheta_0$,
\begin{equation}\label{eqn: QE of gamma}
Q_n(\btheta) -Q_n(\btheta_0)=\frac{1}{2}(\btheta-\btheta_0)'\bA_0(\btheta-\btheta_0)
+\frac{1}{\sqrt{n}}(\btheta-\btheta_0)'{\bf W}_n+O_p(|\btheta-\btheta_0|^3)+o_p(\frac{1}{n}),
\end{equation}
where $ {\bf W}_n= \frac{1}{\sqrt{n}} \sum_i\nabla_1\tau(Y_i, \bX_i;  \btheta_0)$,
$2\bA(\btheta)=E\nabla_2\tau(Y, \bX;  \btheta)$ and $\bA_0=\bA(\btheta_0)$.
Consequently, for the MRC estimator
 $\hat{\btheta}_n$,
\begin{equation}\label{eqn: MRE normality}
 \sqrt{n}(\hat{\btheta}_n-\btheta_0)=A_0^{-1}{\bf W}_n+o_p(1)\convL N({\bf 0}, \bD_0),
\end{equation}
where  $\bD(\btheta)=\bA^{-1} (\btheta)\bV(\btheta) \bA^{-1}(\btheta)$,
$\bV(\btheta)=E(\nabla_1\tau(Y, \bX;  \btheta)[\nabla_1\tau(Y, \bX;  \btheta)]')$ and
$\bD_0=\bD(\btheta_0)$.
\end{proposition}

Because of the standardization,  the rank correlation criterion function $Q_n$ is bounded by
1. It is not difficult to establish a uniform law of large numbers
\begin{equation} \label{eqn: SLLN of Qn}
  \lim_n \sup_{\btheta\in\bTheta}|Q_n(\btheta)-Q(\btheta)|=0, \ a.s.,
\end{equation}
where $Q(\btheta)$ is the expectation of $Q_n(\btheta)$; cf. Han (1987) and Sherman (1993).

\noindent
Likewise, we can show that such uniform convergence also holds for $\tilde{Q}_n$, i.e.
\begin{equation} \label{eqn: SLLN of tQn}
  \lim_n \sup_{\btheta\in\bTheta}|\tilde Q_n(\btheta)-Q(\btheta)|=0, \ a.s.
\end{equation}
Note that the limit $Q$ remains the same.

\smallskip
In the following theorem, we claim that the estimate obtained from maximizing the smoothed
rank correlation function \eqref{eqn: SRC}
is also asymptotically normal with the same asymptotic covariance
matrix as Han's MRCE.

\begin{thm}\label{thy: normality}
  For any given positive definite matrix $\bSigma$, let $ {\widetilde Q_n}(\btheta)$ be defined as in \eqref{eqn: SRC}
and $\widetilde\btheta_n={\rm argmax}_{\btheta} E_Z\Gamma_n(\btheta+Z/\sqrt{n}) {\widetilde Q_n}(\btheta)$.
Then, under Assumptions 1-3, $\widetilde\btheta_n$ is consistent,
$\widetilde\btheta_n\to \btheta_0$ a.s. and asymptotically normal,
$$\sqrt{n}(\widetilde\btheta_n-\btheta_0)
\convL N({\bf 0}, \bD_0),$$
where $\bD_0$ is defined as in Proposition 1.
In addition, $\widetilde\btheta_n$ is asymptotically equivalent to $\hat{\btheta}_n$ in the sense that
$ \widetilde\btheta_n = \hat{\btheta}_n +o_p(n^{-1/2})$.

\end{thm}

Recall that \eqref{eqn: var formula} defines the sandwich-type variance estimator by pretending
that ${\widetilde Q_n}$ is a standard smooth objective function. Theorem \ref{thy: var formula} below
shows that \eqref{eqn: var formula} is consistent.


\begin{thm}\label{thy: var formula}
Let $\hat \btheta_n$ be the MRC estimator and
$\hat{\bD}_n$ be defined by \eqref{eqn: var formula}. Then, for any fixed
positive definite matrix $\bSigma$, $\hat \bD_n(\hat{\btheta}_n,\bSigma)$
converges in probability to $\bD_0$, the limiting variance-covariance matrix of
$\sqrt n (\hat{\btheta}_n-\btheta_0)$.
\end{thm}
\begin{remark}
The self-induced smoothing uses the limiting covariance matrix
$\bD_0$ as $\bSigma$. In practice, we may initially choose the identity matrix for $\bSigma$,
which is the same way as the initial step in algorithm I.
By Theorem 2.1, we know that the one-step estimator $\hat\bSigma_n^{(1)}$ in
algorithm I converges in probability to the true covariance.
However, this one-step estimator depends on the initial choice of $\bSigma$.
Algorithm 1 is an iterative algorithm with the variance-covariance estimator converging to the fixed point of
$\hat{\bD}_n(\hat\btheta_n, \bSigma)=\bSigma $.
\end{remark}

Convergence of Algorithm 1 is ensured by the following theorem.
For notational simplicity, we let $vech(\bB)$ be the vectorization of matrix $\bB$.
For any function $v$ of $\bSigma$,
\[\bigg|\frac{\partial}{\partial \bSigma}\bigg| v =
\sum_{i,j} \bigg|\frac{\partial}{\partial \Sigma_{i,j}} v\bigg|, \quad
\frac{\partial v}{\partial \bSigma}=
(\frac{\partial v}{\partial\Sigma_{1,1}},\frac{\partial v}{\partial\Sigma_{2,1}},...,\frac{\partial v}{\partial\Sigma_{d,d}})',\]
where $\bSigma_{r,s}$ denotes the $(r,s)$ entry of $\bSigma$.

\begin{thm}\label{thy: alg converg}
Let $\hat\bSigma_n^{(k)}$ be defined as in Algorithm 1.
Suppose that Assumptions 1-3 hold. Then there exist $\bSigma^*_n$, $n\geq1$, such that
for any $\epsilon>0$, there exists $N$, such that
for all $n>N$,
\[P(\lim_{k\to\infty} \hat{\bSigma}_n^{(k)}=\bSigma^*_n, \ \ \|\bSigma_n^*-\bD_0\|<\epsilon)
>1-\epsilon.\]
\end{thm}
\begin{remark}
For a fixed $n$,  $\bSigma_n^*$ represents the fixed point matrix in the iterative algorithm. The above theorem
shows that with probability approaching 1, the iterative algorithm converges to a limit, as $k\rightarrow \infty$,
and the limit converges in probability to the limiting covariance matrix $\bD_0$.
\end{remark}
\begin{remark}
The speed of convergence of $\hat\bSigma_n^{(k)}$ to $\bSigma_n^*$ is faster than any exponential rate in the sense that  $\|\hat\bSigma_n^{(k)}-\bSigma_n^*\|=o(\eta^k)$ for any $\eta>0$.
This can be seen from Step 2 of Algorithm 1 in Subsection 2.1 and \eqref{eqn: fast convergence} below,
\begin{equation}\label{eqn: fast convergence}
\sup_{\|\theta-\theta_0\|=o(1),\Sigma\in\mathcal{N}( D_0)}
\bigg|\frac{\partial}{\partial \bSigma}\bigg|[\hat{\bD}_n(\btheta,\bSigma)]_{r,s}=o_p(1),
\end{equation}
which will be proved in the Appendix. Here $\mathcal{N}(\bD_0)$ is a small neighborhood of $\bD_0$
and $\bSigma$ is a positive definite matrix.
\end{remark}


\subsection{Proofs}

In this section, we provide proofs for (1) asymptotic equivalence of
SMRCE to MRCE, (2) consistency of the induced variance estimator and
(3) convergence of Algorithm 1.
Some of the technical developments used in the proofs will be given in the Appendix.


\begin{proof} [Proof of Theorem \ref{thy: normality}]
Without loss of generality, we assume $\btheta_0={\bf 0}.$
As in Subsection 2.1,
let $\bZ$ be a $d$-variate
 normal random vector with mean ${\bf 0}$ and covariance matrix $\bSigma$.
Define
$$ {\widetilde  Q_n}(\btheta)= E_{\bZ} Q_n(\btheta+\bZ/\sqrt{n}).$$
Let
$\Gamma_n(\btheta)=Q_n(\btheta)- Q_n(\btheta_0)$ and
$\widetilde \Gamma_n(\btheta)=E_{\bZ} \Gamma_n(\btheta+\bZ/\sqrt{n})={\widetilde  Q_n}(\btheta)-Q_n(\btheta_0)$.
Define
$$ \widetilde{\btheta}_n={\rm argmax}_{\btheta} \left[ {\widetilde  Q_n}(\btheta)\right]
={\rm argmax}_{\btheta} {\widetilde\Gamma_n}(\btheta).$$
Let $\bOmega_n=I[\|\bZ\|_2>2d {\log n}]$, where $\|\bZ\|_2= \sqrt{\bZ'\bZ}.$
Then $P(\bOmega_n)=o(n^{-2})$ due to the Gaussian tail of $\bZ$.
Since $|Q_n(\btheta)|\leq1$ and $|\Gamma_n(\btheta)|\leq2$,
\[|E_{\bZ} \{\Gamma_n(\btheta+\bZ/\sqrt{n})I[\bOmega_n]\}|\leq  P(\bOmega_n)=o(n^{-2}).\]

\noindent
By the Cauchy-Schwarz inequality,
\[E_{\bZ} \{|\bZ| I[\bOmega_n]\}=o(n^{-2}) \mbox{ and } E_{\bZ} \{|\bZ|^2 I[\bOmega_n]\}=o(n^{-2}).\]

\noindent
By \eqref{eqn: QE of gamma},
uniformly over $o(1)$ neighborhoods of ${\bf 0}$,
\begin{align*}
E_{\bZ} \{\Gamma_n(\btheta &+\bZ/\sqrt{n})I[\bOmega_n^c]\}
=(1/2)E_{\bZ} \{(\btheta+\bZ/\sqrt{n})'\bA_0(\btheta+\bZ/\sqrt{n})I[\bOmega_n^c]\} \\
&+(1/\sqrt{n})E_{\bZ} \{(\btheta+\bZ/\sqrt{n})'{\bf W}_nI[\bOmega_n^c]\}
+o_p(E_{\bZ} \{|\btheta+\bZ/\sqrt{n}|^2I[\bOmega_n^c]\}+\frac{1}{n}).
\end{align*}

Note that
\[E_{\bZ} \{|\btheta+\bZ/\sqrt{n}|^2I[\bOmega_n^c]\} \leq 2
(E_{\bZ}|\btheta|^2+ E_{\bZ}|\bZ|^2/n)=O(|\btheta|^2+1/n).\]

\noindent
Therefore, uniformly over $o(1)$
neighborhoods of ${\bf 0}$,
we have
\begin{equation}\label{eqn: tilde gamma}
\widetilde \Gamma_n(\btheta)=(1/2)\btheta'\bA_0\btheta
+(1/\sqrt{n})\btheta'{\bf W}_n + E(\bZ'\bA_0\bZ)/2n+o_p(|\btheta|^2+1/n).
\end{equation}
Replacing $\btheta$ in \eqref{eqn: tilde gamma} with $\btheta_0={\bf 0}$ and
subtracting it from $\widetilde \Gamma_n(\btheta)$, we have
\begin{equation}
\widetilde\Gamma_n(\btheta)-\widetilde\Gamma_n(\btheta_0) 
=\frac{1}{2}\btheta'\bA_0\btheta
+\frac{1}{\sqrt{n}}\btheta'{\bf W}_n+o_p(|\btheta|^2+1/n).
\end{equation}
Combining (15) with Lemma 1 in the Appendix, we get,
\begin{equation}\label{eqn: SMRE normality}
\sqrt{n}(\widetilde\btheta_n-\btheta_0)=
\bA_0^{-1}{\bf W}_n+o_p(1).
\end{equation}
Therefore, from \eqref{eqn: MRE normality} and \eqref{eqn: SMRE normality},
we have
$$ \sqrt{n}(\hat\btheta_n-\widetilde\btheta_n)=o_p(1).$$
Finally, strong consistency of $\widetilde\btheta_n$ follows the uniform
almost sure convergence of $\tilde Q_n$ as stated in \eqref{eqn: SLLN of tQn}.
This completes the proof.
\end{proof}


\begin{proof}[Proof of Theorem \ref{thy: var formula}]
For notational simplicity, we assume throughout the proof that
$\bSigma$ is the identity matrix. The same argument with modifications
to include constants for up and lower bound may be applied to deal with
a general covariance matrix $\bSigma$.

We first show
\begin{equation}\label{eqn: An conv}
\hat \bA_n(\hat \btheta_n) \convP \bA(\btheta_0).
\end{equation}
\noindent
By definition, $[\hat \bA_n(\btheta)]_{r,s}=\partial^2\tilde Q_n(\btheta)/(\partial\theta_r\partial\theta_s)$.
\smallskip
\noindent
As defined in \eqref{eqn: SRC}$, \tilde Q_n(\btheta)$ has the following integral representation,
\begin{eqnarray*}
\begin{split}
\tilde Q_n(\btheta)=\int Q_n(\btheta+\bz/\sqrt{n})(2\pi)^{-\frac{d}{2}}\exp(-\frac{\|\bz\|_2^2}{2})d\bz.
\end{split}
\end{eqnarray*}
By change of variable $\bt=\btheta+\bz/\sqrt{n}$,
\begin{equation}\label{eqn: Q integral}
\tilde Q_n(\btheta)=\int Q_n(\bt) K_n(\bt,\btheta)d\bt,
\end{equation}

\noindent
$\mbox{where }\displaystyle K_n(\bt,\btheta)=(2\pi)^{-\frac{d}{2}} n^{\frac{d}{2}}\exp(-\frac{n\|\bt-\btheta\|_2^2}{2})$.
From \eqref{eqn: Q integral},
\[\frac{\partial}{\partial \theta_r}\tilde Q_n(\btheta)
=\int Q_n(\bt)\dot{K}_{n,r}(\bt,\btheta)d\bt\]
\noindent
and
\[\frac{\partial^2}{\partial\theta_r\partial\theta_s}\tilde Q_n(\btheta)
=\int Q_n(t)\ddot{K}_{n,r,s}(\bt,\btheta)d\bt,\]
where
$\dot{K}_{n,r}(\bt,\btheta)=\partial K_n(\bt,\btheta)/\partial\theta_r$ and
$\ddot{K}_{n,r,s}(\bt,\btheta)=\partial^2 K_n(\bt,\btheta)/(\partial\theta_r\partial\theta_s)$.

In view of \eqref{eqn: An}, to show \eqref{eqn: An conv}, it suffices to prove
\begin{equation}\label{eqn: Ars}
\int Q_n(\bt)\ddot{K}_{n,r,s}(\bt,\btheta)d\bt = [\bA(\btheta_0)]_{r,s} + o_p(1)\end{equation}
uniformly over $\|\btheta-\btheta_0\|=O(n^{-1/2})$.
To show \eqref{eqn: Ars}, we define
\[\bOmega_{n,r}=\left\{\bt: (t_r-\theta_r)^2<\frac{4\log n}{n},
\sum_{i \neq r}(t_i-\theta_i)^2 < \frac{2(d-1)\log n}{n}\right\}.\]

By Lemma 2(i) and the boundedness of $Q_n(\bt)$, we have,
\[ \int_{(\bOmega_{n,r}\cap\bOmega_{n,s})^c} Q_n(\bt)\ddot{K}_{n,r,s}(\bt,\btheta)d\bt=o(n^{-1/2}),\]
where ${\mathfrak B}^c$ for set ${\mathfrak B}$ denotes its complement.
Therefore, \eqref{eqn: Ars} reduces to
\begin{equation}\label{eqn: Arsp}
\int_{\bOmega_{n,r}\cap\bOmega_{n,s}} Q_n(\bt)\ddot{K}_{n,r,s}(\bt,\btheta)d\bt = [\bA(\btheta_0)]_{r,s} + o_p(1).
\end{equation}

To show \eqref{eqn: Arsp}, we establish a quadratic expansion of $Q_n(\bt)$ for
$\bt \in \bOmega_{n,r}\cap\bOmega_{n,s}$.
Since $\|\bt-\btheta\|_2< \sqrt{4d\log n/n}$ for $\bt \in \bOmega_{n,r}\cap\bOmega_{n,s}$
and $\|\btheta-\btheta_0\|_2=O(n^{-1/2})$, it follows that $\|\bt-\btheta_0\|_2=o(1)$.
Therefore, by \eqref{eqn: QE of gamma},
\begin{eqnarray}\label{eqn: QE of Qn}
\begin{split}
Q_n(\bt)&=
Q_n(\btheta_0) + \frac{1}{2}(\bt-\btheta_0)'\bA(\btheta_0)(\bt-\btheta_0)\\
&+(\bt-\btheta_0)'{\bf W}_n/\sqrt{n}+O_p(|\bt-\btheta_0|^3) + o_p(1/n).
\end{split}
\end{eqnarray}
Therefore, the left hand side of \eqref{eqn: Arsp} equals $\bf{I}+ \bf{II}+ \bf{III}+ \bf{IV}$, where
\begin{eqnarray*}
\begin{split}
\bf{I}&=\int_{\bOmega_{n,r}\cap\bOmega_{n,s}}\left[O_p(|\bt-\btheta_0|^3) + o_p(1/n)\right]\ddot{K}_{n,r,s}(\bt,\btheta)d\bt,\\
\bf{II}&=Q_n(\btheta_0)\times\int_{\bOmega_{n,r}\cap\bOmega_{n,s}}\ddot{K}_{n,r,s}(\bt,\btheta)d\bt,\\
\bf{III}&=\frac{{\bf W}_n'}{\sqrt{n}}\times\int_{\bOmega_{n,r}\cap\bOmega_{n,s}}(\bt-\btheta_0)\ddot{K}_{n,r,s}(\bt,\btheta)d\bt,\\
\bf{IV}&=\frac{1}{2}\int_{\bOmega_{n,r}\cap\bOmega_{n,s}}(\bt-\btheta_0)'\bA(\btheta_0)(\bt-\btheta_0)\ddot{K}_{n,r,s}(\bt,\btheta)d\bt.
\end{split}
\end{eqnarray*}

By the definition of $\bOmega_{n,r}$,
\[|{\bf I}|\leq\bigg|O_p\left(\frac{(\log n)^{\frac{3}{2}}}{n\sqrt{n}}\right)+o_p(1/n)\bigg| \times
\int_{\bOmega_{n,r}\cap\bOmega_{n,s}} |\ddot{K}_{n,r,r}(\bt,\btheta)|d\bt.\]
By Lemma 2(ii), ${\bf I}=o_p(1)$. Furthermore,
$\displaystyle {\bf II}=o(n^{-1/2})$ due to Lemma 2(iii).
Note that
\begin{eqnarray}\label{eqn: split III}
\begin{split}
\int_{\bOmega_{n,r}\cap\bOmega_{n,s}} (\bt-\btheta_0)\ddot{K}_{n,r,s}(\bt,\btheta)d\bt & =
\int_{\bOmega_{n,r}\cap\bOmega_{n,s}} (\bt-\btheta)\ddot{K}_{n,r,s}(\bt,\btheta)d\bt \\
& + (\btheta-\btheta_0)\int_{\bOmega_{n,r}\cap\bOmega_{n,s}} \ddot{K}_{n,r,s}(\bt,\btheta)d\bt\\
& = (\btheta-\btheta_0)\int_{\bOmega_{n,r}\cap\bOmega_{n,s}} \ddot{K}_{n,r,s}(\bt,\btheta)d\bt,
\end{split}
\end{eqnarray}
where the last equality follows from the fact that $\bOmega_{n,r}$ and $\bOmega_{n,s}$
are symmetric at $\btheta$ and $(\bt-\btheta)\ddot{K}_{n,r,s}(\bt,\btheta)$ is
an odd function of $[\bt-\btheta]_{r}$ for $r=1,2,...,d$.
Combining this with Lemma 2(i), we have $\bf{III}=o(n^{-1})$.
Again by symmetry,
\begin{eqnarray}\label{eqn: split IV}
\begin{split}
\int_{\bOmega_{n,r}\cap\bOmega_{n,s}} & (\bt-\btheta_0)'\bA(\btheta_0)(\bt-\btheta_0)\ddot{K}_{n,r,s}(\bt,\btheta) d\bt \\
& = \int_{\bOmega_{n,r}\cap\bOmega_{n,s}}(\bt-\btheta)'\bA(\btheta_0)(\bt-\btheta)\ddot{K}_{n,r,s}(\bt,\btheta) d\bt \\
& + (\btheta-\btheta_0)'\bA(\btheta_0)(\btheta-\btheta_0)\int_{\bOmega_{n,r}\cap\bOmega_{n,s}}\ddot{K}_{n,r,s}(\bt,\btheta) d\bt.
\end{split}
\end{eqnarray}
By Lemma 2 (i) and (iv), ${\bf IV}=[\bA(\btheta_0)]_{r,s} + o(n^{-1/2})$.
Combining the approximations for $\bf{I} - \bf{IV}$,
we get \eqref{eqn: Arsp}.

\smallskip
Next we prove $\hat \bV_n(\hat \theta_n) \convP  \bV(\theta_0)$ by showing, componentwise,
\begin{equation}\label{eqn: Vrs}
[\hat \bV_n(\btheta)]_{r,s} = [\bV(\btheta_0)]_{r,s} + o_p(1)
\end{equation}
uniformly over $\|\btheta-\btheta_0\|=O(n^{-1/2})$ for $r,s=1,...,d$.

Define

\centerline{$q(\bu, \tilde \bu; \btheta)=
I_{[y>\tilde y]}I_{[(\bx-\tilde \bx)'\bbeta>0]} + I_{[y<\tilde y]}I_{[(\bx-\tilde \bx)'\bbeta<0]},$}

\smallskip
\noindent
where $\bu=(y,\bx)$ and $\tilde \bu=(\tilde y, \tilde \bx)$. In addition, let
$\displaystyle \tau_n(\bu,\btheta)=
\int q(\bu, \tilde \bu; \btheta)\mathbb{F}_n(d\tilde \bu),$
where $\mathbb{F}_n(\cdot)$ is the empirical distribution for $\bu_i$'s.
By definition,
\begin{eqnarray*}
\begin{split}
[\hat \bV_n(\btheta)]_{r,s}
& =\frac{1}{n}\sum_{i=1}^n
\left[\frac{\partial}{\partial\theta_r} \int \tau_n(\bu_i,\btheta+\frac{\bz}{\sqrt{n}})
(2\pi)^{-\frac{d}{2}}e^{-\frac{\|\bz\|_2^2}{2}}d\bz\right]\\
& \quad \quad \times
\left[\frac{\partial}{\partial\theta_s} \int \tau_n(\bu_i,\btheta+\frac{\tilde \bz}{\sqrt{n}})
(2\pi)^{-\frac{d}{2}}e^{-\frac{\|\tilde \bz\|_2^2}{2}}{d\tilde\bz}\right].
\end{split}
\end{eqnarray*}

\noindent
Letting $\bt=\btheta+\bz/\sqrt{n}$ and $\bomega=\btheta+\tilde \bz/\sqrt{n}$, we have
\begin{eqnarray*}
\begin{split}
[\hat \bV_n(\btheta)]_{r,s}
& =\frac{1}{n}\sum_{i=1}^n
\frac{\partial}{\partial\theta_r} \int \tau_n(\bu_i,\bt)K_n(\bt,\btheta)d\bt
\times  \frac{\partial}{\partial\theta_r} \int \tau_n(\bu_i,\bomega) K_n(\bomega,\theta)d\bomega\\
& = \frac{1}{n}\sum_{i=1}^n
\int \tau_n(\bu_i,\bt)\dot{K}_{n,r}(\bt,\btheta)d\bt \times
    \int \tau_n(\bu_i,\bomega)\dot{K}_{n,s}(\bomega,\btheta)d\bomega\\
&=\int G_n(\bt,\bomega) \dot{K}_{n,r}(\bt,\btheta)\dot{K}_{n,s}(\bomega,\btheta)d\bt d\bomega,
\end{split}
\end{eqnarray*} where
$G_n(\bt,\bomega) = \frac{1}{n}\sum_{i=1}^n\tau_n(\bu_i,\bt)\tau_n(\bu_i,\bomega)$,
which is bounded by 0 and 1.
By Lemma 2 (vii),
\begin{eqnarray}\label{eqn: Vnp}
\begin{split}
[\hat \bV_n(\btheta)]_{r,s}&=o(n^{-\frac{1}{2}})\\
&+\int_{\bOmega_{n,r}\times\bOmega_{n,s}}
G_n(\bt,\bomega) \dot{K}_{n,r}(\bt,\btheta)\dot{K}_{n,s}(\bomega,\btheta)d\bt d\bomega
\end{split}
\end{eqnarray}
uniformly over $\|\btheta-\btheta_0\|=O(n^{-\frac{1}{2}})$. Let
$f(\bu,\bv,\bw;\btheta_1,\btheta_2)=q(\bu,\bv;\btheta_1)\times q(\bu,\bw;\btheta_2)$ and
$f^*(\bu,\bv,\bw;\btheta_1,\btheta_2)$, the symmetrized $f$.
By definition,
\begin{eqnarray}\label{eqn: Gnp}
\begin{split}
G_n(\btheta_1,\btheta_2)&=\frac{1}{{n \choose 3}}\sum_{i<j<k} f^*(\bu_i,\bu_j,\bu_k;\btheta_1,\btheta_2)\\
    &+\frac{1}{n}\times\frac{1}{{n \choose 2}}\sum_{i<j} f^*(\bu_i,\bu_j,\bu_j;\btheta_1,\btheta_2)
    \triangleq U_n + \frac{1}{n} \tilde U_n.
\end{split}
\end{eqnarray}
Clearly $U_n$ is a third-order U-statistics and $\tilde U_n$ is a second-order U-statistics.
Applying Hoeffding's decomposition (van der Vaart, 1998, section 12.3),
\begin{equation}\label{eqn: Un Hoeffding}
U_{n} = \sum_{c=0}^3 {{3}\choose{c}} U_{n,c},
\end{equation}
where $U_{n,c}$ is a U-statistics of order $c$ ($c=0,1,2,3$) and defined as
\[U_{n,c}=\frac{1}{{{3}\choose{c}}}\sum_{|B|=c} \frac{1}{{{n}\choose{3}}}
\sum_{i} P_B\left[f^*(\bu_{i_1},\bu_{i_2},\bu_{i_3};\btheta_1,\btheta_2)\right].\]
Here, adopting the notations from van der Vaart (1998, Section 11.4), we define
$P_B\left[f^*(\bu_{i_1},\bu_{i_2},\bu_{i_3};\btheta_1,\btheta_2)\right]$
as a projection of $f^*$ such that
\begin{eqnarray*}
\begin{split}
P_{{\o}} f^* &= E f^*, \\
P_{\{i\}} f^* & = E [f^* | \bu_i] - E f^*, \\
P_{\{i,j\}} f^* & = E [f^* | \bu_i, \bu_j] - E [f^* | \bu_i] - E [f^* | \bu_j] + E f^*,\\
P_{\{1,2,3\}} f^* & = E [f^* | \bu_1, \bu_2, \bu_3] - \sum_{i\neq j} E [f^* | \bu_i, \bu_j] + \sum_{i=1,2,3}E [f^* | \bu_i] - E f^*.
\end{split}
\end{eqnarray*}
We know from Hoeffding's decomposition that
$U_{n,2}$ and $U_{n,3}$ are second- and third-order degenerated U-statistics
with bounded kernels
and thus of order $o_p(n^{-1})$ and $o_p(n^{-3/2})$; see Sherman (1992, Corollary 8).
Therefore, by Lemma 2(vi),
\begin{equation}\label{eqn: integrating Un}
\int_{\bOmega_{n,r}\times\bOmega_{n,s}} U_{n,c}
\dot{K}_{n,r}(\bt,\btheta)\dot{K}_{n,s}(\bomega,\btheta)d\bt d\bomega=o_p(1), \mbox{ for } c=2,3.
\end{equation}
Replacing $U_{n,c}$ by $\tilde U_n/n$ in \eqref{eqn: integrating Un} also results in $o_p(1)$.
Then combining this and \eqref{eqn: integrating Un} with \eqref{eqn: Gnp} and \eqref{eqn: Un Hoeffding},
\eqref{eqn: Vnp} reduces to
\begin{equation}\label{eqn: Vnpp}
\begin{split}
[\hat \bV_n(\btheta)]_{r,s} &= 3 \times\int_{\bOmega_{n,r}\cap\bOmega_{n,s}}
    U_{n,1}\times\dot{K}_{n,r}(\bt,\btheta)\dot{K}_{n,s}(\bomega, \btheta)d\bt d\bomega \\
&+\int_{\bOmega_{n,r}\cap\bOmega_{n,s}} Ef\times\dot{K}_{n,r}(\bt,\btheta)\dot{K}_{n,s}(\bomega,\btheta)d\bt d\bomega + o_p(1).
\end{split}
\end{equation}
Let $f_1(\bu_j;\bt, \bomega)=E[f(\bu,\bv,\bw;\bt,\bomega)|\bu=\bu_j]$,
$f_2(\bv_j;\bt, \bomega)=E[f(\bu,\bv,\bw;\bt,\bomega)|\bv=\bv_j]$ and
$f_3(\bw_j;\bt, \bomega)=E[f(\bu,\bv,\bw;\bt,\bomega)|\bw=\bw_j]$. We define
$\displaystyle \tilde G_n(\bt,\bomega) = \frac{1}{n}\sum_{j=1}^n f_1(\bu_j;\bt, \bomega)$.
By the definitions of $f(\bu,\bv,\bw;\bt,\bomega)$ and $q(\bu,\bv;\btheta)$,
we have $\displaystyle \tilde G_n(\bt,\bomega)=\frac{1}{n}\sum_{i=1}^n \tau(\bu_i,\bt)\tau(\bu_i,\bomega)$.
By Lemma 3 and applying integration by parts twice,
\begin{eqnarray*}
\begin{split}
\int_{\bOmega_{n,r} \times \bOmega_{n,s}}
& \tilde G_n(\bt,\bomega)\dot{K}_{n,r}(\bt,\btheta)\dot{K}_{n,s}(\bomega,\btheta)d\bt d\bomega= o(n^{-1/2})\\
+ & \int_{\tilde\bOmega_{n,r} \times \tilde\bOmega_{n,s}} \left\{\frac{1}{n}\sum_{i=1}^n
\frac{\partial\tau(\bu_i,\btheta+\frac{\bz}{\sqrt{n}})}{\partial\theta_r}
\frac{\partial\tau(\bu_i,\btheta+\frac{\tilde \bz}{\sqrt{n}})}{\partial\theta_s}\right\}
\prod_i d\Phi(z_i)d\Phi(\tilde z_i) ,
\end{split}
\end{eqnarray*}
where $\tilde\bOmega_{n,r}:= \{\bz: z_r^2<4\log n, \sum_{i \neq r}z_i^2 < 2(d-1)\log n\}$.
By Lemma 3,

\[\frac{1}{n}\sum_{i=1}^n\frac{\partial}{\partial\theta_r}\tau(\bu_i,\btheta_1^*)\frac{\partial}{\partial\theta_s} \tau(\bu_i,\btheta_2^*)=[\bV(\btheta_0)]_{r,s}+o_p(1)\]
uniformly over $\{(\btheta_1^*, \btheta_2^*): \|\btheta_i^*-\btheta_0\|_2=o(1), i=1,2\}$.
Therefore,
\[\int_{\bOmega_{n,r} \times \bOmega_{n,s}} \frac{1}{n}\sum_{j=1}^n f_1(\bu_j;\bt,\bomega)\times
\dot{K}_{n,r}(\bt,\btheta)\dot{K}_{n,s}(\bomega,\btheta)d\bt d\bomega = [\bV(\btheta_0)]_{r,s}+o_p(1).\]

\noindent
Similarly, applying integration by parts and by Lemma 3 and 2(vi), we have
\[\int_{\bOmega_{n,r} \times \bOmega_{n,s}} \frac{1}{n}\sum_{j=1}^n f_2(\bv_j;\bt,\bomega)\times
\dot{K}_{n,r}(\bt,\btheta)\dot{K}_{n,s}(\bomega,\btheta)d\bt d\bomega = [\bV(\btheta_0)]_{r,s}+o_p(1),\]
\[\int_{\bOmega_{n,r} \times \bOmega_{n,s}} \frac{1}{n}\sum_{j=1}^n f_3(\bw_j;\bt,\bomega)\times
\dot{K}_{n,r}(\bt,\btheta)\dot{K}_{n,s}(\bomega,\btheta)d\bt d\bomega = [\bV(\btheta_0)]_{r,s}+o_p(1),\]
\[\int_{\bOmega_{n,r} \times \bOmega_{n,s}} E[f(\bu,\bv,\bw;\bt,\bomega)]\times
\dot{K}_{n,r}(\bt,\btheta)\dot{K}_{n,s}(\bomega,\btheta)d\bt d\bomega = [\bV(\btheta_0)]_{r,s}+o_p(1).\]
Hence the right hand side of \eqref{eqn: Vnpp} is $[\bV(\btheta_0)]_{r,s}+o_p(1)$, which gives \eqref{eqn: Vrs}.
From \eqref{eqn: An conv} and \eqref{eqn: Vrs},
$\hat \bD_n(\hat\btheta_n)\convP  \bD_0$.
\end{proof}


\begin{proof}[Proof of Theorem \ref{thy: alg converg}]
From Theorem \eqref{thy: var formula}, we know that
$\hat{\bSigma}_n^{(1)}\convP  \bD_0$ and
 $\hat\bSigma_n(\hat\btheta_n, \bD_0) \convP  \bD_0$.
By the mean value theorem,
\begin{eqnarray*}
\begin{split}
[\hat{\bSigma}_n^{(2)} - \bD_0]_{r,s}
&= [\hat\bSigma_n(\hat\btheta_n, \hat{\bSigma}_n^{(1)}) - \hat\bSigma_n(\hat\btheta_n, \bD_0)]_{r,s}
    + [\hat\bSigma_n(\hat\btheta_n, \bD_0) - \bD_0]_{r,s}\\
&= \left[\frac{\partial}{\partial\bSigma}[\hat{\bD}_n]_{r,s}\bigg|_{\bSigma=\bSigma^*}\right]'
    \times vech(\hat{\bSigma}_n^{(1)}-\bD_0)
+ [\hat\bSigma_n(\hat\btheta_n, \bD_0) - \bD_0]_{r,s},
\end{split}
\end{eqnarray*}
where $\|\bSigma^*-\bD_0\|\leq\|\hat{\bSigma}_n^{(1)}-\bD_0\|$ and thus $\bSigma^* \in \mathcal{N}
( \bD_0)$.
In view of Lemma 4 and $\hat\bSigma_n(\hat\btheta_n, \bD_0) \convP  \bD_0$,
$\hat{\bSigma}_n^{(2)} \convP  \bD_0$.
Again by the mean value theorem,
$$[\hat{\bSigma}_n^{(k+1)}-\hat{\bSigma}_n^{(1)}]_{r,s}=
\left[\frac{\partial}{\partial \bSigma}[\hat{ \bD}_n]_{r,s}\bigg|_{\bSigma=\bSigma^*}\right]'
\times vech(\hat{\bSigma}_n^{(k)}-\bD_0),$$
where $\|\bSigma^*-\bD_0\|\leq\|\hat{\bSigma}_n^{(k)}-\bD_0\|$. Then by Lemma 4 and
mathematical induction, we know that for any $\epsilon>0$ and $\eta>0$, there exist $K$ and $N$,
such that for any $n>N$ and $k>K$,
$$P\left(\big|[\hat\bSigma_n^{(k+1)}-\hat\bSigma_n^{(k)}]_{r,s}\big| \leq \eta \times
\big|[\hat\bSigma_n^{(k)}-\hat\bSigma_n^{(k-1)}]_{r,s}\big|, \mbox{ for all } k>K\right)>1-\epsilon,$$
where $1\leq s,r\leq d$.
Note that the inequality inside the above probability implies that
$\hat{\bSigma}_n^{(k)}$ converges as $k\rightarrow\infty$ and the limit $\bSigma_n^*$
satisfies $\bSigma_n^*=\hat{ \bD}_n(\hat\btheta_n,\bSigma_n^*)$ and $\bSigma_n^* \convP \bD_0$.
\end{proof}


\section{Extensions}

In this section, we extend the approach to
the partial rank correlation (PRC)
criterion function $Q_n^*$, defined by (3), of Khan and Tamer (2007) for censored data.
Under the usual conditional independence between failure and censoring times given covariates and
additional regularity conditions, Khan and Tamer (2007)
developed asymptotic properties for PRCE that are parallel to those by Sherman (1993).

The same self-induced smoothing can be applied to partial rank correlation criteria function
to get
\setcounter{equation}{28}

\begin{eqnarray}\label{eqn: SPRC}
\begin{split}
\tilde{Q}^*_n(\btheta)&=E_{\bZ} Q^*_n(\btheta+\bZ/\sqrt{n})\\
&=\frac{1}{n(n-1)}\sum_{i\neq j}
\Delta_j I[\tilde Y_i > \tilde Y_j]\Phi\left(\sqrt{n}\bX_{ij}'\bbeta(\btheta)/\sigma_{ij}\right).
\end{split}
\end{eqnarray}
We define its maximizer, $\tilde{\btheta}^*_n$,
as the smoothed partial rank correlation estimator (SPRCE).
Let
\begin{equation}
\quad \hat{\bA}^*_n(\btheta,\bSigma)
=\frac{1}{2n(n-1)}\sum_{i \neq j}\left\{ H_{ij}\times\dot{\phi}\left(\frac{\sqrt{n}\bX_{ij}'\bbeta}{\sigma_{ij}}\right)
\left[\frac{\sqrt{n}\bX_{ij}^{(1)}}{\sigma_{ij}}\right]^{\otimes2}\right\},
\end{equation}
\begin{eqnarray}
\begin{split}
\quad\quad \hat{\bV}^*_n(\btheta, \bSigma)
=\frac{1}{n^3}\sum_{i=1}^n\left\{\sum_{j} \left[H_{ij}\times\phi\left(\frac{\sqrt{n}\bX_{ij}'\bbeta}{\sigma_{ij}}\right)
\frac{\sqrt{n}\bX_{ij}^{(1)}}{\sigma_{ij}}\right]\right\}^{\otimes2}, \\
\end{split}
\end{eqnarray}
\begin{equation}\label{eqn: SPRCE var}
\hat \bD_n^*(\btheta,\bSigma)=
[\hat \bA^*_n(\btheta, \bSigma)]^{-1}\times
\hat \bV_n^*(\btheta, \bSigma)\times
[\hat \bA^*_n(\btheta, \bSigma)]^{-1},\end{equation}
where $H_{ij}=\Delta_j \times I[\tilde Y_i > \tilde Y_j]-\Delta_i \times I[\tilde Y_j > \tilde Y_i]$.

\smallskip
Based on $\hat \bD_n^*(\btheta,\bSigma)$, we have the following iterative algorithm
to compute the SPRCE and variance estimate simultaneously.

\begin{proc} (SPRCE)
\begin{enumerate}{\leftmargin=-1em}
\item Compute the PRC estimator $\hat{\btheta}^*_n$ and set $\hat\bSigma^{(0)}$ to be the identity matrix.

\item Update variance-covariance matrix $\hat{\bSigma}_n^{*(k)}=\hat \bD_n^*(\hat{\btheta}^*_n,\hat{\bSigma}_n^{*(k-1)})$.
Smooth the partial rank correlation $Q^*_n(\btheta)$
using covariance matrix $\hat{\bSigma}_n^{*(k)}$. Maximize the resulting smoothed partial rank correlation to get
an estimator $\hat{\btheta}_n^{*(k)}$.

\item Repeat step 2 until $\hat{\btheta}_n^{*(k)}$ converge.
\end{enumerate}
\end{proc}

In addition to Assumptions 1-3, Khan and Tamer (2007) added the following
assumption for the consistency of PRCE.
\begin{assumption}
Let $\bS_X$ be the support of $\bX_i$, and $\mathfrak{X}_{uc}$ be the set
\[\mathfrak{X}_{uc}=\{\bx\in \bS_X: P(\Delta_i=1|\bX_i=\bx)>0\}.\]
Then $P(\mathfrak{X}_{uc})>0$.
\end{assumption}

\smallskip
Similar to the rank correlation function, it can be shown that
under Assumptions 1-4, \eqref{eqn: QE of gamma} and
\eqref{eqn: SLLN of Qn}
 still hold for partial rank correlation function
$Q_n^*(\btheta) $.
Therefore, Theorems 1-3 in Section 2 continue to hold
when replacing
the point and variance estimators for smoothed rank correlation
by the corresponding ones for the
smoothed partial rank correlation.
Specifically, for any positive definite matrix $\bSigma$,
under Assumptions 1-4, we have

\begin{enumerate}
\item The SPRCE $\widetilde\btheta_n^*$ is asymptotically equivalent to the PRCE $\hat{\btheta}_n^*$
in the sense that
$ \widetilde\btheta_n^* = \hat{\btheta}_n^* +o_p(n^{-1/2})$, and, therefore,
\[\sqrt{n}(\widetilde\btheta_n^*-\theta_0)\convL N({\bf 0}, \bD_0^*),\]
where $\bD_0^*$ is the limiting variance-covariance matrix of
$\hat{\btheta}_n^*$.

\item Variance estimator is consistent: $\hat \bD_n^*(\hat{\btheta}_n^*,\bSigma) \convP \bD_0^*$.

\item Algorithm 2 converges numerically in the sense that
there exist $\bSigma^*_n$, $n\geq1$, such that
for any $\epsilon>0$, there exists $N$, such that
for all $n>N$,
$P(\lim_{k\to\infty} \hat{\bSigma}_n^{*(k)}=\bSigma^*_n, \ \ \|\bSigma_n^*-\bD_0^*\|<\epsilon)
>1-\epsilon.$
\end{enumerate}

\smallskip
The proofs are similar to
those of Theorems 1-3 in Section 2, and are, therefore, omitted.

\section{Numerical results}
In this section, we first apply the proposed self-induced smoothing method to analyze
the primary biliary cirrhosis (PBC) data (Fleming and Harrington, 1990, Appendix D)
and compare the result with that using the Cox regression. We then report results from several
simulation studies we conducted using the method.

\subsection{PBC data}
We applied smoothed PRCE to the survival times of the first 312 subjects with no missing covariates
in the PBC data. We included two covariates albumin and age50 (age divided by 50).
We reparameterized the transformation model (1)
by setting $\beta_{age50}$ as 1, and estimated $\theta_{albumin}$ by SPRCE.
We also calculated PRCE for $\theta_{albumin}$ and fitted the standard Cox model.
For the Cox regression, the ratio $\hat\beta_{albumin}/\hat\beta_{age50}$ is the estimate of $\theta_{albumin}$.
The results are summarized in Table 1.
\begin{center}{Table 1\\Regression Analysis of PBC data}\end{center}
        \begin{center}
            \begin{tabular}{cccccc}

                  & Albumin & SE \\
            \hline\hline
            SPRCE & -4.29 & 1.40 \\

            PRCE  & -3.50 & - \\

            Cox  & -3.04 & 0.60 \\

            \end{tabular}
        \end{center}

\medskip
Note that PRCE does not have a readily available standard error estimate.
The standard error of $\hat\beta_{albumin}/\hat\beta_{age50}$ in the Cox model
was estimated by the delta method.
Estimates from both the SPRCE and the Cox model conclude that the ratio of
$\beta_{albumin}$ to $\beta_{age50}$ is significant.

\begin{figure}[h]
\centering
\includegraphics[width=4.5in]{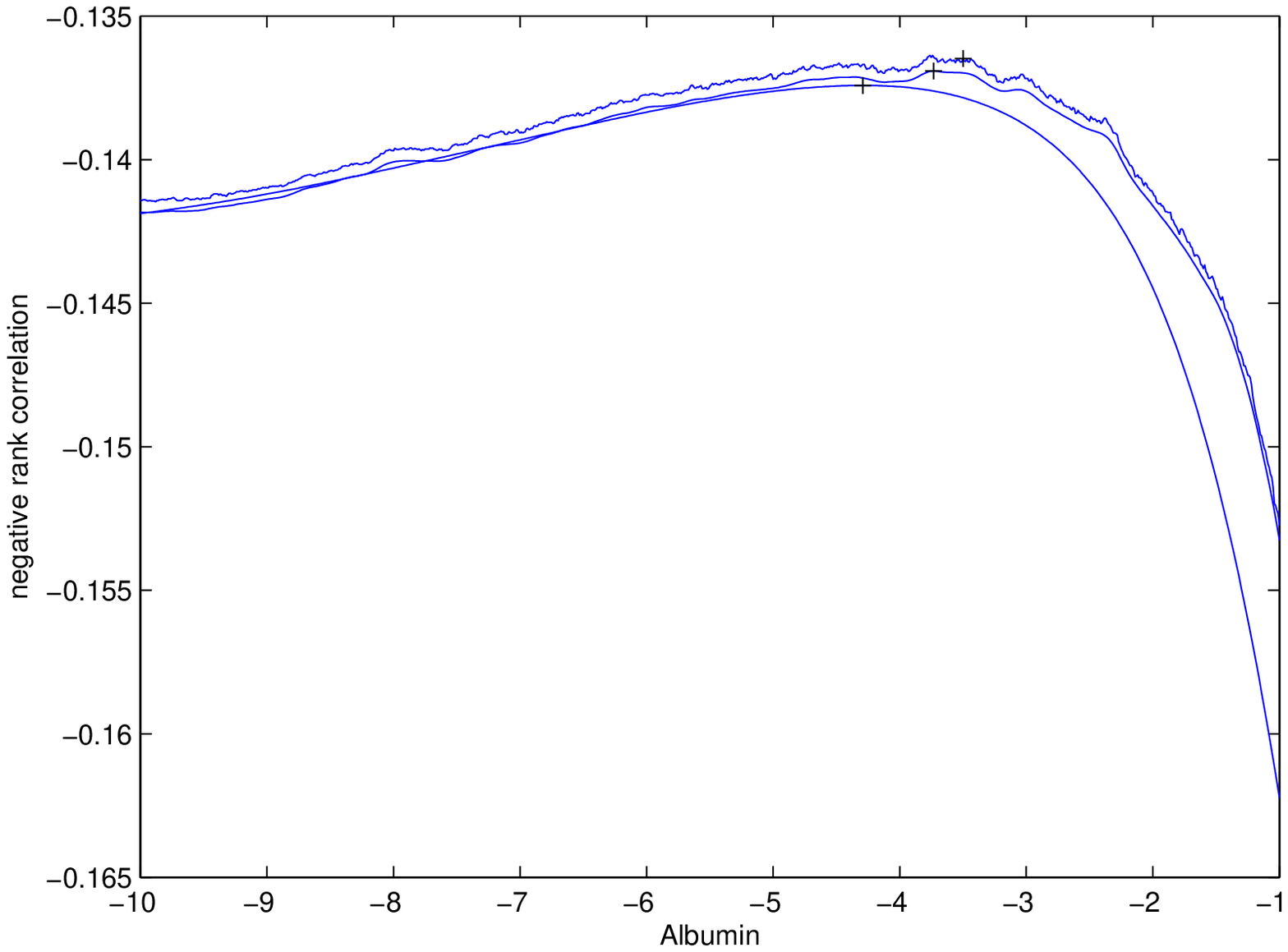}
\caption{}
\end{figure}

To further assess the self-induced smoothing procedure, we plot
the original objective function as well as the smoothed
one in the first and last steps of our algorithm, as shown in Figure 1.
The top curve is the original objective function, the middle curve
is one after the initial smoothing, and the bottom curve
is the limit of the iterative algorithm (after 8 iterations).
It appears that the one-step smoothed objective function is under-smoothed in terms of
the level of fluctuations, and the limiting curve is quite smooth.

\begin{center}{Table 2\\The proportional hazard model without censoring}\end{center}
        \begin{center}
            \begin{tabular}{ccccccc}

            $n=500$ & Est & Mean & Bias & RMSE & SE & coverage\\
            \hline\hline
            $\theta$   & SMRCE & 1.601 & $1.2\times10^{-3}$ & 0.0298 & 0.0316 & 92.3$\%$ \\

                        & MRCE  & 1.601 & $0.8\times10^{-3}$ & 0.0340 & - & - \\

                        & Cox  & 1.599 & $-0.9\times10^{-3}$ & 0.0200 & - & - \\

            \end{tabular}
        \end{center}
        \begin{center}
            \begin{tabular}{ccccccc}

            $n=1000$& Est & Mean & Bias & RMSE & SE & coverage\\
            \hline\hline
            $\theta$   & SMRCE & 1.601 & $1.0\times10^{-3}$ & 0.0193 & 0.0212 & 93.9$\%$ \\

                        & MRCE  & 1.600 & $0.2\times10^{-3}$ & 0.0225 & - & - \\

                        & Cox  & 1.600 & $0.1\times10^{-3}$ & 0.0141 & - & - \\

            \end{tabular}
        \end{center}
        \begin{center}
            \begin{tabular}{ccccccc}

            $n=2000$& Est & Mean & Bias & RMSE & SE & coverage\\
            \hline\hline
            $\theta$   & SMRCE & 1.600 & $0.2\times10^{-3}$ & 0.0136 & 0.0144 & 94.9$\%$ \\

                        & MRCE  & 1.600 & $-0.1\times10^{-3}$ & 0.0158 & - & - \\

                        & Cox  & 1.600 & $0.1\times10^{-3}$ & 0.0100 & - & - \\

            \end{tabular}
        \end{center}
\medskip
\begin{center}{Table 3\\The proportional hazard model with censoring}\end{center}
        \begin{center}
            \begin{tabular}{ccccccc}

            $n=600$ & Est & Mean & Bias & RMSE & SE & coverage\\
            \hline\hline
            $\theta$   & SPRCE & 1.604 & $3.7\times10^{-3}$ & 0.0282 & 0.0300 & 93.2$\%$ \\

                        & PRCE  & 1.603 & $2.9\times10^{-3}$ & 0.0327 & - & - \\

                        & Cox  & 1.601 & $1.0\times10^{-3}$ & 0.0204 & - & - \\

            \end{tabular}
        \end{center}
        \begin{center}
            \begin{tabular}{ccccccc}

            $n=1200$& Est & Mean & Bias & RMSE & SE & coverage\\
            \hline\hline
            $\theta$   & SPRCE & 1.601 & $1.1\times10^{-3}$ & 0.0190 & 0.0201 & 93.9$\%$ \\

                        & PRCE  & 1.601 & $0.8\times10^{-3}$ & 0.0217 & - & - \\

                        & Cox  & 1.600 & $-0.2\times10^{-3}$ & 0.0139 & - & - \\

            \end{tabular}
        \end{center}
        \begin{center}
            \begin{tabular}{ccccccc}

            $n=2400$& Est & Mean & Bias & RMSE & SE & coverage\\
            \hline\hline
            $\theta$ & SPRCE & 1.600 & $0.4\times10^{-3}$ & 0.0127 & 0.0136 & 95.4$\%$ \\

            & PRCE  & 1.600 & $0.1\times10^{-3}$ & 0.0148 & - & - \\

            & Cox  & 1.600 & $-0.2\times10^{-3}$ & 0.0097 & - & - \\

            \end{tabular}
        \end{center}

\newpage
\begin{center}{Table 4\\The linear model with gaussian noise}\end{center}
        \begin{center}
            \begin{tabular}{c c c c c c c }

            $n=250$ & Est & Mean & Bias & RMSE & SE & coverage\\
            \hline\hline
            $\theta_1$   & SMRCE & 1.615 & $1.5\times10^{-2}$ & 0.0747 & 0.0756 & 91.7$\%$ \\

                        & MRCE  & 1.612 & $1.2\times10^{-2}$ & 0.0730 & - & -\\

                        & LS  & 1.601 & $ 0.7\times10^{-3}$ & 0.0296 & - & -\\
            \hline
            $\theta_2$   & SMRCE & .5042 & $0.4\times10^{-2}$ & 0.0427 & 0.0443 & 93.6$\%$ \\

                        & MRCE  & .5058 & $0.5\times10^{-2}$ & 0.0423 & - & -\\

                        & LS  & .5006 & $ 0.6\times10^{-3}$ & 0.0354 & - & -\\

            \end{tabular}
        \end{center}
                \begin{center}
            \begin{tabular}{c c c c c c c }

             $n=500$ & Est & Mean & Bias & RMSE & SE & coverage\\
            \hline\hline
            $\theta_1$   & SMRCE & 1.605 & $4.9\times10^{-3}$ & 0.0515 & 0.0513 & 92.7$\%$ \\

                        & MRCE  & 1.607 & $6.7\times10^{-3}$ & 0.0523 & - & -\\

                        & LS  & 1.601 & $ 0.7\times10^{-3}$ & 0.021 & - & -\\
            \hline
            $\theta_2$   & SMRCE & .5023 & $2.3\times10^{-3}$ & 0.0296 & 0.0302 & 94.6$\%$ \\

                        & MRCE  & .5042 & $4.2\times10^{-3}$ & 0.0316 & - & -\\

                        & LS  & .5006 & $0.6\times10^{-3}$ & 0.0254 & - & -\\

            \end{tabular}
        \end{center}
                \begin{center}
            \begin{tabular}{c c c c c c c }

             $n=1000$ & Est & Mean & Bias & RMSE & SE & coverage\\
            \hline\hline
            $\theta_1$   & SMRCE & 1.603 & $3.6\times10^{-3}$ & 0.0361 & 0.0348 & 92.4$\%$ \\

                        & MRCE  & 1.603 & $3.4\times10^{-3}$ & 0.0382 & - & -\\

                        & LS  & 1.601 & $0.5\times10^{-3}$ & 0.0144 & - & -\\
            \hline
            $\theta_2$   & SMRCE & .5009 & $0.9\times10^{-3}$ & 0.0203 & 0.0207 & 94.8$\%$ \\

                        & MRCE  & .5018 & $1.8\times10^{-3}$ & 0.0214 & - & -\\

                        & LS  & .5004 & $0.4\times10^{-3}$ & 0.0176 & - & -\\

            \end{tabular}
        \end{center}

\subsection{Simulation studies}

We conducted simulation studies for a number of cases. In the first case (Design I),
we generated $\bX$ from a bivariate normal distribution with mean $[-10, 20]'$ and a covariance matrix diag$\{3^2, 2^2\}$.
We set $\bbeta_0^T=(\theta, 1)=[1.6, 1]$ and generated $\epsilon$ from the probability density function
$f(w)=2\exp(2w-\exp(2w))$.
We set the transformation $H(x)$ as $H^{-1}(y)=log(y^2)$. This is indeed a Weibull proportional hazard model.
The sample sizes were $n=500,1000,2000$ and the numbers of replications were $500$.
The SMRCE, MRCE and Cox model were used to estimate $\theta$,
and the standard error of SMRCE was computed by Algorithm 1.
The mean(Mean), bias(Bias) and root mean square error(RMSE) for
each method as well as mean of standard error(SE) and
coverage of $95\%$ confidence interval for the SMRCE are reported in Table 2.

The second case (Design II) is similar to the first one except
that $Y$ is censored by a random variable $C$, which is independent of $\bX$ and normally distributed with
mean $\mu=9.2$ and variance $\sigma^2=0.5^2$. The sample sizes were $n=600,1200,2400$ and the numbers of replications were $500$.
This design is similar to that in G{\o}rgens and Horowitz (1999).
The SPRCE, PRCE and Cox model were used to estimate $\theta$, and
the standard error of SPRCE was computed by Algorithm 2.
The resulting estimates are summarized in Table 3 where we also report
bias(Bias), root mean square error(RMSE), mean of standard error(SE), and coverage of $95\%$ confidence interval.

In the third case (Design III), we generated $\bX=[X_1, X_2, X_3]'$ by two steps. We first
generated $[X_1,X_3]'$ from a bivariate normal distribution with mean $\left[-2,2\right]'$ and an identity covariance matrix.
We then generated $X_3$ as 0 or 2 with equal probability. We set $\bbeta_0^T=(\theta_1,\theta_2, 1)=[1.6, 0.5, 1]$
and generated $\epsilon$ from a normal distribution with $\mu=0$ and $\sigma^2=0.5^2$.
We set the transformation $H(x)=x$. The sample sizes were $n=250,500,1000$ and the numbers of replications were $500$.
The SMRCE, MRCE and least squared method were applied to estimate $\theta_1$ and $\theta_2$,
and the standard error of SMRCE was computed by Algorithm 1.
Table 4  reports
the mean (Mean), bias (Bias) and root mean square error (RMSE) for each method
as well as mean of standard error (SE) and coverage of $95\%$ confidence interval for the SMRCE.

From Tables 2, 3 and 4, we find that (1) the root mean squared error is close to
the mean standard error for the SMRCE (SPRCE); (2) as the sample size
increases, the bias reduces and the coverage of $95\%$ confidence interval converges to
the nominal level. These show that the proposed variance estimator is accurate and Algorithms 1 and 2 work well.

\section{Discussion}
This paper provides a simple yet general recipe for smoothing the
discontinuous rank correlation criteria function. The smoothing is
self-induced in the sense that the implied bandwidth is
essentially the asymptotic standard deviation of the regression
parameter estimator. It is shown that such smoothing does not
introduce any significant bias in that the resulting estimator is
asymptotically equivalent to the original maximum rank correlation
estimator, which is asymptotically normal. The smoothed rank
correlation can be used as if it were a regular smooth criterion
function in the usual M-estimation problem, in the sense that the
standard sandwich-type plug-in variance-covariance estimator is
consistent. Simulation and real data analysis provide additional
evidence that the proposed method gives the right amount of
smoothing.

Because of the family of transformation models contains both the
proportional hazards and accelerated failure time models as its
submodels, the new approach may be used for model selection. The
specification test commonly used in the econometrics literature
(Hausman, 1978) may also be used for testing a specific
semiparametric model assumption. In addition, the smoothed
objective function also makes it possible to fit a penalized
regression by introducing a LASSO-type penalty.

The method and theory developed herein can easily be extended to
other problems of similar nature, i.e. discontinous objective
functions with associated estimators being asymptotically normal.
In particular, we can apply the self-induced smoothing to the
estimator introduced by Chen (2002) for the transformation
function $H$ to obtain a consistent variance estimator.

\appendix

\numberwithin{equation}{section}

\section{Lemmas, Corollaries and Proofs}

Lemma 1 below is due to Sherman (1993, Theorem 2).
\begin{lemma} We denote $\Gamma_n(\btheta)$ as general
objective functions which are centered and
satisfies the same regularity conditions as in Sherman (1993).
Suppose
$\btheta_n:=argmax_{\bTheta}\Gamma_n(\btheta)$
 is  consistent for $\btheta_0$, an
interior point of $\bTheta$. Suppose also that uniformly over $o_p(1)$
neighborhoods of $\btheta_0$,

$\displaystyle \Gamma_n(\btheta)=\frac{1}{2}(\btheta-\btheta_0)'\bA(\btheta-\btheta_0)
+\frac{1}{\sqrt{n}}(\btheta-\btheta_0)'{\bf W}_n+o_p(|\btheta-\btheta_0|^2)+o_p(1/n)$
\smallskip
where $\bA$ is a negative definite matrix, and ${\bf W}_n$ converges
in distribution to a $N({\bf 0}, \bV)$ random vector. Then
$$\sqrt{n}(\btheta_n-\btheta_0)=
-\bA^{-1}{\bf W}_n+o_p(1)\convL N({\bf 0}, \bA^{-1}\bV \bA^{-1}).$$
\end{lemma}

Recall in Theorem 2, we define
$K_n(\bt,\btheta)=(2\pi)^{-\frac{d}{2}} n^{\frac{d}{2}}\exp(-\frac{n\|\bt-\btheta\|_2^2}{2})$
and its first and second partial derivatives with respect to $\btheta$ as
\begin{eqnarray*}
\begin{split}
& \dot{K}_{n,r}(\bt,\btheta)
=(\frac{2\pi}{n})^{-\frac{d}{2}} n(t_r-\theta_r)e^{-\frac{n\|\bt-\tiny{\theta}\|_2^2}{2}},\\
& \ddot{K}_{n,r,r}(\bt,\btheta)
=(\frac{2\pi}{n})^{-\frac{d}{2}} n(n(t_r-\theta_r)^2-1)e^{-\frac{n\|\bt-\tiny{\theta}\|_2^2}{2}},\\
& \ddot{K}_{n,r,s}(\bt,\btheta)
=(\frac{2\pi}{n})^{-\frac{d}{2}} n^2(t_r-\theta_r)(t_s-\theta_s)e^{-\frac{n\|\bt-\tiny{\theta}\|_2^2}{2}}.
\end{split}
\end{eqnarray*}
Also recall
\[\bOmega_{n,r}=\{\bt: (t_r-\theta_r)^2<4\log n/n,
\sum_{i \neq r}(t_i-\theta_i)^2 < 2(d-1)\log n /n\}.\]
Then we have the following lemma.

\begin{lemma}
Uniformly over $\|\btheta-\btheta_0\|_2=O(n^{-\frac{1}{2}})$,

\noindent
(i) \quad $\displaystyle
\int_{(\bOmega_{n,r}\cap\bOmega_{n,s})^c}F(\bt)\ddot{K}_{n,r,s}(\bt,\btheta)d\bt=o(\frac{1}{\sqrt{n}}),
\forall F(\bt) \mbox{ s.t. }0\leq F(\bt)\leq 1.$\\

\noindent
(ii)\quad $\displaystyle \int_{\bOmega_{n,r}\cap\bOmega_{n,s}} \frac{1}{n}|\ddot{K}_{n,r,s}(\bt,\btheta)|d\bt=O(1)$.\\

\noindent
(iii)\quad $\displaystyle \int_{\bOmega_{n,r}\cap\bOmega_{n,s}} \ddot{K}_{n,r,s}(\bt,\btheta)d\bt=o(n^{-1/2})$.\\

\noindent
(iv) \quad $\displaystyle \int_{\bOmega_{n,r}\cap\bOmega_{n,s}}
\frac{1}{2}(\bt-\btheta)'\bA(\bt-\btheta) \ddot{K}_{n,r,s}(\bt,\btheta)d\bt=[A]_{r,s}+o(n^{-1}).$\\

\noindent
(v)\quad $\displaystyle \mbox{ }\int_{\bOmega_{n,r}^c} |\dot{K}_{n,r}(\bt,\btheta)|d\bt=O(n^{-3/2})$.\\

\noindent
(vi)\quad $\displaystyle \mbox{ }\int_{\bOmega_{n,r}} \frac{1}{\sqrt{n}}|\dot{K}_{n,r}(\bt,\btheta)|d\bt=O(1)$.\\

\noindent
(vii)\quad $\mbox{ For any given }0\leq G(\bt,\bomega)\leq 1, $\\

$\mbox{ }\mbox{ }\mbox{ }\mbox{ }
\displaystyle \int G(\bt,\bomega) \dot{K}_{n,r}(\bt,\btheta)\dot{K}_{n,s}(\bomega,\btheta)d\bt d\bomega$\\
$\mbox{ }\mbox{ }\mbox{ }\mbox{ }\mbox{ }\mbox{ }\mbox{ }\mbox{ }\mbox{ }\mbox{ }\mbox{ }\mbox{ }
\displaystyle =\int_{\bOmega_{n,r}\times\bOmega_{n,s}}
G(\bt,\bomega) \dot{K}_{n,r}(\bt,\btheta)\dot{K}_{n,s}(\bomega,\btheta)d\bt d\bomega + o(n^{-\frac{1}{2}}).$
\end{lemma}

\begin{proof} Let
$\tilde\bOmega_{n,r}=\left\{\bt: t_r^2<4\log n/n,
\sum_{i \neq r}t_i^2 < 2(d-1)\log n /n\right\}$,
and divide its complement into
$\tilde \bOmega_{n,r}^{(1)}:=\Big\{\bt: t_r^2>4\log n/n\Big\}$ and
$\tilde \bOmega_{n,r}^{(2)}:=\left\{\bt: t_r^2<4\log n/n,
\sum_{i \neq r}t_i^2 \geq 2(d-1)\log n /n\right\}$.
We prove (i)-(iv) for $s=r$. For $s \neq r$,
the proofs are similar and omitted.

\medskip
For (i), note that
\[\int_{\bOmega_{n,r}^c} F(\bt)\ddot{K}_{n,r,r}(\bt,\btheta)d\bt
=\int_{\tilde\bOmega_{n,r}^c} F(\bt+\btheta)\ddot{K}_{n,r,r}(\bt,{\textbf 0})d\bt.\]

\smallskip
\noindent
Since $0\leq F(\bt)\leq 1$ and $(nt_r^2-1)I[\tilde \bOmega_{n,r}^{(1)}]\geq0$,
\begin{eqnarray*}
\begin{split}
\int_{\tilde\bOmega_{n,r}^{(1)}} F(\bt+\btheta)\ddot{K}_{n,r,r}(\bt,{\textbf 0}) &d\bt
= \int_{\tilde\bOmega_{n,r}^{(1)}} F(\bt+\btheta) (\frac{2\pi}{n})^{-\frac{d}{2}}
n(n t_r^2-1)e^{-\frac{n\|\bt\|_2^2}{2}}d\bt\\
\leq & (\frac{2\pi}{n})^{-\frac{d}{2}} \int_{\tilde\bOmega_{n,r}^{(1)}}
n(n t_r^2-1)e^{-\frac{n\|\bt\|_2^2}{2}}d\bt\\
= & (2\pi)^{-\frac{1}{2}} \int_{\tilde\bOmega_{n,r}^{(1)}} n^{\frac{1}{2}}
d(n t_r e^{-\frac{n t_r^2}{2}}) \prod_{i \neq r} d\Phi(\sqrt{n}t_i)=o(\frac{1}{\sqrt{n}}),
\end{split}
\end{eqnarray*}
where the last equality follows from $\displaystyle 0\leq \int \prod_{i \neq r} d\Phi(\sqrt{n}t_i)\leq 1$. Similarly,
\begin{eqnarray*}
\begin{split}
&\int_{\tilde\bOmega_{n,r}^{(2)}} F(\bt+\btheta)\ddot{K}_{n,r,r}(\bt,{\textbf 0}) d\bt
= \int_{\tilde\bOmega_{n,r}^{(2)}} F(\bt+\btheta) (\frac{2\pi}{n})^{-\frac{d}{2}}
(n^2t_r^2-n)e^{-\frac{n\|\bt\|_2^2}{2}}d\bt      \\
\leq & (2\pi)^{-\frac{1}{2}} \int_{\tilde\bOmega_{n,r}^{(2)}} n |n t_r^2-1 |
e^{-\frac{n t_r^2}{2}}d \sqrt{n}t_r \prod_{i \neq r} d\Phi(\sqrt{n}t_i)
\leq 8\frac{n (\log n)^{3/2}}{n^2}=o(\frac{1}{\sqrt{n}}).
\end{split}
\end{eqnarray*}

For (ii), by definition,
\begin{eqnarray*}\begin{split}
\frac{1}{n}\int_{\bOmega_{n,r}} &|\ddot{K}_{n,r,r}(\bt,\btheta)|d\bt
= (2\pi)^{-\frac{1}{2}} \int_{\tilde\bOmega_{n,r}}
\sqrt{n}|nt_r^2-1| e^{-\frac{nt_r^2}{2}}d t_r \prod_{i \neq r} d\Phi(\sqrt{n}t_i)   \\
\leq & \frac{1}{\pi \sqrt{n}}\left(\left[n t_r e^{-\frac{nt_r^2}{2}}\right]
\bigg|_{t_r=1/\sqrt{n}}^{2\sqrt{\frac{\log n}{n}}}
+ \left[n t_r e^{-\frac{nt_r^2}{2}}\right]
\bigg|_{t_r=1/\sqrt{n}}^{0}\right) = O(1),
\end{split}\end{eqnarray*}
where the inequality follows from $\displaystyle 0\leq \int \prod_{i \neq r} d\Phi(\sqrt{n}t_i) \leq 1$.

For (iii), by definition,
$\displaystyle \int_{\bOmega_{n,r}} \ddot{K}_{n,r,r}(\bt,\btheta)d\bt
= \int_{\tilde\bOmega_{n,r}} (\frac{2\pi}{n})^{-\frac{d}{2}}
(n^2 t_r^2-n)e^{-\frac{n \|t\|_2^2}{2}}d\bt$ \\
$\displaystyle = 2 n^{\frac{3}{2}}t_r e^{-\frac{n t_r^2}{2}}
\bigg|_{t_r=0}^{2\sqrt{\frac{\log n}{n}}}\times \int_{\tilde\bOmega_{n,r}} \prod_{i \neq r} d\Phi(\sqrt{n}t_i)=o(\frac{1}{\sqrt{n}})$,
where the last equality follows from $\displaystyle 0\leq \int \prod_{i \neq r} d\Phi(\sqrt{n}t_i)\leq 1$.\\

For (iv), by definition and applying integration by parts twice,
\begin{eqnarray*}
\begin{split}
& \int_{\bOmega_{n,r}} \frac{1}{2}(\bt-\btheta)'\bA(\bt-\btheta) \ddot{K}_{n,r,r}(\bt,\btheta)d\bt\\
= & \int_{\tilde\bOmega_{n,r}} \frac{1}{2}\bt'\bA\bt
(2\pi)^{-\frac{1}{2}} \sqrt{n} d(-n t_r e^{-\frac{n t_r^2}{2}})
\prod_{i \neq r}d\Phi(\sqrt{n}t_i)                                   \\
= & o(n^{-1}) + \int_{\tilde\bOmega_{n,r}} \bt'\bA {\bf e_r}
(2\pi)^{-\frac{1}{2}} \sqrt{n} d(-e^{-\frac{n t_r^2}{2}})
\prod_{i \neq r}d\Phi(\sqrt{n}t_i)                                          \\
= & \bA_{r,r} + o(n^{-1}),
\end{split}
\end{eqnarray*}
where ${\bf e_r}'=(0,...,0,1,0,...,0)$ with $r^{th}$ entry being 1,
and the last equality follows from the Gaussian tail probability.

For (v), we know, by definition,
$\displaystyle \int_{\bOmega_{n,r}^C} |\dot{K}_{n,r}(\bt,\btheta)|d\bt
= \int_{\tilde\bOmega_{n,r}^C} |\dot{K}_{n,r}(\bt,{\textbf 0})|d\bt$.
By symmetry,
\begin{eqnarray*}
\begin{split}
\int_{\tilde\bOmega_{n,r}^{(1)}} |\dot{K}_{n,r}(\bt,{\textbf 0})| &d\bt
= \frac{2}{\sqrt{2\pi}}\int_{t_r\geq0, \tilde\bOmega_{n,r}^{(1)}} n^{\frac{1}{2}}
d(e^{-\frac{n t_r^2}{2}}) \prod_{i \neq r} d\Phi(\sqrt{n}t_i)\\
&\leq\frac{\sqrt{n}}{n^2}\times1=
O(n^{-\frac{3}{2}}),
\end{split}
\end{eqnarray*}
where the inequality follows from
$\displaystyle 0\leq \int \prod_{i \neq r} d\Phi(\sqrt{n}t_i)\leq 1$.
Similarly,
\begin{eqnarray*}
\begin{split}
\int_{\tilde\bOmega_{n,r}^{(2)}} |\ddot{K}_{n,r,r}(\bt,{\textbf 0})| &d\bt
= \frac{2\sqrt{n}}{\sqrt{2\pi}} \int_{t_r\geq0, \tilde\bOmega_{n,r}^{(2)}}
d(e^{-\frac{n t_r^2}{2}}) \prod_{i \neq r} d\Phi(\sqrt{n}t_i)
\leq\frac{\sqrt{n}}{n^2}=O(n^{-\frac{3}{2}}).
\end{split}
\end{eqnarray*}

For (vi), by definition,
\begin{eqnarray*}
\begin{split}
&\frac{1}{\sqrt{n}}\int_{\bOmega_{n,r}} |\dot{K}_{n,r}(\bt,\btheta)| d\bt
= \int_{\tilde\bOmega_{n,r}} (2\pi)^{-\frac{d}{2}} n^{\frac{d-1}{2}}
n|t_r|e^{-\frac{n\|t\|_2^2}{2}}d\bt\\
= & (2\pi)^{-\frac{1}{2}} 2\int_{t_r\geq0, \tilde\bOmega_{n,r}}
d(e^{-\frac{n t_r^2}{2}}) \prod_{i \neq r} d\Phi(\sqrt{n}t_i)=O(1),
\end{split}
\end{eqnarray*}
where the second equality is due to symmetry and the third equality
follows from $\displaystyle 0\leq \int \prod_{i \neq r} d\Phi(\sqrt{n}t_i)\leq 1$.

To prove (vii), without loss of generality, we assume $0\leq G(\bt,\bomega)\leq 1$.
We denote $\bOmega_{n,r}^a$ as $\bOmega_{n,r}$ and $\bOmega_{n,r}^b$ its complement. Then,
\begin{eqnarray*}
\begin{split}
\int_{\bOmega_{n,r}^k \times \bOmega_{n,s}^l}
& G(\bt,\bomega) \dot{K}_{n,r}(\bt,\btheta)\dot{K}_{n,s}(\bomega,\btheta)d\bt d\bomega\\
& \leq \int_{\bOmega_{n,r}^k} |\dot{K}_{n,r}(\bt,\btheta)|d\bt \times
 \int_{\bOmega_{n,s}^l}|\dot{K}_{n,s}(\bomega,\btheta)|d\bomega \\
& = \int_{\tilde\bOmega_{n,r}^k} |\dot{K}_{n,r}(\bt,{\textbf 0})|d\bt
\times \int_{\tilde\bOmega_{n,s}^l}|\dot{K}_{n,s}(\bomega,{\textbf 0})|d\bomega
\end{split}
\end{eqnarray*}
where $k$ and $l$ are chosen from $\{a,b\}$. Then by (v) and (vi),
\begin{eqnarray*}
\begin{split}
\int G(\bt,\bomega) & \dot{K}_{n,r}(\bt,\btheta)\dot{K}_{n,s}(\bomega,\btheta)d\bt d\bomega\\
&=\int_{\bOmega_{n,r}\times\bOmega_{n,s}}
G(\bt,\bomega) \dot{K}_{n,r}(\bt,\btheta)\dot{K}_{n,s}(\bomega,\btheta)d\bt d\bomega + o(n^{-\frac{1}{2}}).
\end{split}
\end{eqnarray*}
\end{proof}

\begin{lemma} Uniformly over $(\bt,\bomega)$ such that
$\|\bt-\btheta_0\|=o(1)$ and $\|\bomega-\btheta_0\|=o(1)$, we have
\begin{equation}\label{eqn: lemma 3(1)}
\frac{1}{n}\sum_{i=1}^n \left[\frac{\partial\tau(\bu_i,\bt)}{\partial\theta_r}
 \tau(\bu_i,\bomega)\right]
=E\left[\frac{\partial\tau(\bu,\btheta_0)}{\partial\theta_r} \tau(\bu,\btheta_0)\right]+o_p(1),
\end{equation}
\begin{equation}\label{eqn: lemma 3(2)}
\frac{1}{n}\sum_{i=1}^n \left[\frac{\partial\tau(\bu_i,\bt)}{\partial\theta_r}
\frac{\partial \tau(\bu_i,\bomega)}{\partial\theta_s}\right]
=E\left[\frac{\partial\tau(\bu,\btheta_0)}{\partial\theta_r}
\frac{\partial\tau(\bu,\btheta_0)}{\partial\theta_s}\right]+o_p(1).\end{equation}
\end{lemma}

\begin{proof} We sketch the main steps of the proof below.

First of all, observe that
\begin{eqnarray*}
\begin{split}
\bigg|\frac{1}{n}\sum_{i=1}^n
\left[\frac{\partial\tau(\bu_i,\bt)}{\partial\theta_r}\tau(\bu_i,\bomega)\right]
&-\frac{1}{n}\sum_{i=1}^n
\left[\frac{\partial\tau(\bu_i,\btheta_0)}{\partial\theta_r}\tau(\bu_i,\bomega)\right]\bigg|          \\
&\leq\frac{1}{n}\sum_{i=1}^n
\left[\bigg|\frac{\partial\tau(\bu_i,\bt)}{\partial\theta_r}
-\frac{\partial\tau(\bu_i,\btheta_0)}{\partial\theta_r}\bigg| \times \big|\tau(\bu_i,\bomega)\big|\right]    \\
&\leq \frac{1}{n}\sum_{i=1}^n M_2(\bu_i)\times|\bt-\btheta_0|,
\end{split}
\end{eqnarray*}
where $M_2(\bu)$ is an integrable function.
The last inequality is due to Assumption 3 and $|\tau(\bu,\btheta)|\leq 1$.

Since $M_2(\bu)$ is integrable, by the law of large numbers, the left hand side of above inequality is thus $o_p(1)$.
By a similar argument, we can show that\\
\[\frac{1}{n}\sum_{i=1}^n
\left[\frac{\partial\tau(\bu_i,\bt)}{\partial\theta_r} \tau(\bu_i,\bomega)\right]=
\frac{1}{n}\sum_{i=1}^n
\left[\frac{\partial\tau(\bu_i,\btheta_0)}{\partial\theta_r} \tau(\bu_i,\btheta_0)\right]+o_p(1).\]
By the law of large numbers, we get \eqref{eqn: lemma 3(1)}.
The proof of \eqref{eqn: lemma 3(2)} is similar.
\end{proof}

\begin{lemma} Let $\hat{\bA}_n$ and $\hat{\bV}_n$ be the same as those in \eqref{eqn: var formula}.
Then, for $1 \leq r,s \leq d$, we have
\[
\sup_{\|\theta-\theta_0\|=o(1),\Sigma\in\mathcal{N}( D_0)}
\bigg|\frac{\partial}{\partial \bSigma}\bigg|
[\hat{\bA}_n(\btheta,\bSigma)]_{r,s}=o_p(1),\]
\[\sup_{\|\theta-\theta_0\|=o(1),\Sigma\in\mathcal{N}( D_0)}
\bigg|\frac{\partial}{\partial \bSigma}\bigg|
[\hat{\bV}_n(\btheta,\bSigma)]_{r,s}=o_p(1),\]
where $\mathcal{N}(\bD_0)$ is a small neighborhood of $\bD_0$ and $\bSigma$
is a positive definite matrix.
\end{lemma}

\begin{proof}
We now extend the definition of kernels in Lemma 2 for any covariance matrix $\bSigma$ as follows,
\[K_n(\bt,\btheta,\bSigma):=(\frac{2\pi}{n})^{-\frac{d}{2}}|\bSigma|^{-1/2}
\exp(-\frac{n}{2}(\bt-\btheta)'\bSigma^{-1}(\bt-\btheta)),\]
where $|\bSigma|$ is the determinant of $\bSigma$.
Then the first and second derivatives of $K_n$ with respect to $\btheta$ become
\begin{eqnarray*}
\begin{split}
& \dot{K}_{n,r}(\bt,\btheta,\bSigma)
:=(\frac{2\pi}{n})^{-\frac{d}{2}}|\bSigma|^{-\frac{1}{2}} n{\bf e_r}'\bSigma^{-1}(\bt-\btheta)
e^{-\frac{n}{2}(\bt-\theta)'\bSigma^{-1}(\bt-\theta)},\\
& \ddot{K}_{n,r,s}(\bt,\btheta,\bSigma) :=(\frac{2\pi}{n})^{-\frac{d}{2}}|\bSigma|^{-\frac{1}{2}} n
{\bf e_r}'\left[n\bSigma^{-1}(\bt-\btheta)(\bt-\btheta)'\bSigma^{-1}-\bSigma^{-1}\right]{\bf e_s}\\
& \qquad\qquad\qquad\qquad\qquad\times e^{-\frac{n}{2}(\bt-\theta)'\bSigma^{-1}(\bt-\theta)}.
\end{split}
\end{eqnarray*}

We partition $\bRe^d$ into $\bOmega_{n,r}$ and its complement $\bOmega_{n,r}^c$, where
$\bOmega_{n,r}:=\Big\{\bt: (\bt-\btheta)'\bSigma^{-1}(\bt-\btheta)<6d\log n /n\Big\}$. Furthermore,
We define $\tilde\bOmega_{n,r}:=\Big\{\bt: \bt'\bSigma^{-1}\bt<6d\log n /n\Big\}.$

\smallskip
Note that $(\bt-\btheta)(\bt-\btheta)' e^{-\frac{1}{2}(\bt-\theta)'\bSigma^{-1}(\bt-\theta)}$ is bounded
for $\bSigma\in\mathcal{N}( D_0)$.
Similar to the proofs of Theorem 2 and Lemma 2, we can get
$$\frac{\partial [\hat \bA_n(\btheta,\bSigma)]_{r,s}}{\partial \bSigma}=
[\bA(\btheta_0)]_{r,s}\int_{\tilde\bOmega_{n,r}} \frac{\partial K_n(\bt,{\textbf 0},\bSigma)}{\partial \bSigma}d\bt+o_p(1),$$
$$\frac{\partial [\hat \bV_n(\btheta,\bSigma)]_{r,s}}{\partial \bSigma}=
[\bV(\btheta_0)]_{r,s}\int_{\tilde\bOmega_{n,r}} \frac{\partial K_n(\bt,{\textbf 0},\bSigma)}{\partial \bSigma}d\bt+o_p(1),$$
uniformly over $(\btheta,\bSigma)$ such that $\|\btheta-\btheta_0\|=o(1)$ and $\|\bSigma-\bD_0\|=o(1)$.

Likewise, we have
$\displaystyle \int_{\tilde\bOmega_{n,r}^c} \frac{\partial K_n(\bt,{\textbf 0},\bSigma)}{\partial \bSigma}d\bt=o(1)$,
which, combined with
$\displaystyle \int \frac{\partial K_n(\bt,{\textbf 0},\bSigma)}{\partial \bSigma}d\bt=0$, implies
$\displaystyle \int_{\tilde\bOmega_{n,r}} \frac{\partial K_n(\bt,{\textbf 0},\bSigma)}{\partial \bSigma}d\bt=o(1).$
This completes the proof.
\end{proof}

\begin{corollary} For $1 \leq r,s \leq d$, we have\\

\centerline{$\displaystyle
\sup_{\|\theta-\theta_0\|=o(1),\Sigma\in\mathcal{N}( D_0)}
\bigg|\frac{\partial}{\partial \bSigma}\bigg|
[\hat{\bA}_n(\btheta,\bSigma)^{-1}]_{r,s}=o_p(1)$.}
\end{corollary}

\begin{proof}
First, by Theorem 2, Lemma 4 and the mean value theorem, we can show that
$[\hat{\bA}_n(\btheta,\bSigma)]_{r,s}
=[\hat{\bA}_n(\btheta, \bSigma)-\hat{\bA}_n(\btheta,  \bD_0)
+\hat{\bA}_n(\btheta,  \bD_0)]_{r,s}=[\bA(\btheta_0)]_{r,s}+o_p(1)$.
By matrix differentiation, $d \bA^{-1} = -\bA^{-1}(d\bA)\bA^{-1}$.
Thus $\hat \bA_n^{-1}-\bA_0^{-1}=-\bA_0^{-1}(\hat \bA_n-\bA_0)\bA_0^{-1}+o(\| \hat \bA_n-\bA_0 \|_1)$,
where $\bA_0=\bA(\btheta_0)$.
The rest of the proof is straightforward and thus omitted.
\end{proof}

\begin{lemma}\label{thy: fast convergence}
For $1\leq r, s \leq d$, we have
\[\sup_{\|\theta-\theta_0\|=o(1),\Sigma\in\mathcal{N}( D_0)}
\bigg|\frac{\partial}{\partial \bSigma}\bigg|[\hat{ \bD}_n(\btheta,\bSigma)]_{r,s}=o_p(1).\]
\end{lemma}

\begin{proof}
The result follows immediate from Lemma 4 and Corollary 1.
\end{proof}

\section{A sufficient condition for Assumption 3}


Suppose $f$ is the joint density for $(\bX,Y)$ and $f(\cdot|\bdr,s)$ is the conditional density of $X^{(2)}$ given $\bX^{(1)}=\bdr$ and $Y=s$. Suppose $g(\cdot|s,\btheta)$ is the conditional density of $\bX'\bbeta(\btheta)$ given $Y=s$ and $g(\cdot|\bdr,s,\btheta)$ is the conditional density of $\bX'\bbeta(\btheta)$ given $\bX^{(1)}=\bdr$ and $Y=s$. By change of variable, $g(\bt|\bdr,s,\btheta)=f(\bt-\bdr'\btheta|\bdr,s)$. Therefore,
\begin{center}
$g(\bt|s,\btheta)=\int g(\bt|\bdr,s,\btheta)G_{\bX^{(1)}|s}(d\bdr)=\int f(\bt-\bdr'\btheta|\bdr,s)G_{\bX^{(1)}|s}(d\bdr)$,
\end{center}
where $G_{\bX^{(1)}|s}$ is the conditional distribution of $\bX^{(1)}$ given $Y=s$. We also observe that,
\begin{center}
$\tau(z,\btheta)=\int_{-\infty}^{x'\beta(\theta)} \int_{-\infty}^{y} g(\bt|s,\btheta)G_Y(ds)d\bt +
\int_{x'\beta(\theta)}^{\infty} \int_{y}^{\infty} g(\bt|s,\btheta)G_Y(ds)d\bt$,
\end{center}
where $G_Y$ is the marginal distribution of $Y$.
Therefore if the conditional density $f_{\bX^{(2)}|\bX^{(1)},Y}(\cdot|\bdr,s)$ has bounded derivatives up to order three for each $(\bdr,s)$ in the support of space $\bX^{(1)}\otimes Y$, it is not difficult to show that Assumption 3 is satisfied.
 The sufficient condition can be easily verified in certain common situations such as
 when the conditional density $f_{\bX^{(2)}|\bX^{(1)},Y}$ is normal.


\begin{thebibliography}{9}






\bibitem{r1} \textsc{ Bennett, S.} (1983). Analysis of
survival data by the proportional odds model. {\it Statist. Med.},
{\bf 2}, 273-277.

\bibitem{r1} \textsc{ Box, G. E. P. and Cox, D. R.} (1964). An analysis of transformations.
{\it J. R. Statist. Soc. B}, {\bf 26}, 211-252.


\bibitem{r1} \textsc{ Brown, B. M. and Wang, Y.} (2005). Standard errors and
covariance matrices for smoothed rank estimators. {\it Biometrika},
{\bf 92}, 149-158.


\bibitem{r1} \textsc{ Brown, B. M. and Wang, Y.} (2007). Induced smoothing for rank regression with censored survival times. {\it Statist. Med.},
{\bf 26}, 828-836.

\bibitem{r1}
\textsc{ Chen, S.} (2002). Rank estimation of transformation models.
{\it Econometrica}, {\bf 70}, 1683-1697.

\bibitem{r1} \textsc{ Cox, D. R.} (1972).
Regression models and life tables.
{\it J. R. Statist. Soc. B}, {\bf 34}, 187-220.

\bibitem{r1} \textsc{ Cox, D. R. and Oakes, D.} (1984). {\it Analysis of Survival Data.} London: Chapman and Hall.

\bibitem{r1}
\textsc{ Efron, B.} (1979). Bootstrap Methods: Another Look at the Jackknife.
{\it Ann. Statist}, {\bf 7}, 1-26.


\bibitem{r1}
\textsc{ G{\o}rgens, T. and Horowitz, J. L.} (1999). Semiparametric
estimation of a censored regression model with an unknown
transformation of the dependent variable. {\it J. Econometrics},
{\bf 90}, 155-191.


\bibitem{r1}
\textsc{Han, A. K.} (1987). Non-parametric analysis of a generalized
regression model.
{\it J. Econometrics}, {\bf 35}, 303-316.

\bibitem{r1}
\textsc{Hausman, J. A.} (1978). Specification tests in econometrics.
{\it Econometrica}, {\bf 46}, 1251-1271.

\bibitem{r1}
\textsc{Jin, Z., Ying, Z. and Wei, L. J.} (2001). A simple resampling method by perturbing the minimand.
{\it Biometrika}, {\bf 88}, 381-390.


\bibitem{r1} \textsc{Kalbfleisch, J. D. and Prentice, R. L.} (2002). {\it The Statistical Analysis of Failure Time Data (2nd Ed.).}
New York: Wiley.


\bibitem{r1}
\textsc{Khan, S. and Tamer, E.} (2007). Partial rank estimation of duration models with general forms of censoring.
{\it J. Econometrics}, {\bf 136}, 251-280.


\bibitem{r1}
\textsc{Maddala, G. S.} (1983). Limited dependent and qualitative variables in econometrics.
{\it Econometric Society Monograph No.3}. Cambridge: Cambridge University Press.


\bibitem{r1}
\textsc{McFadden, D. L.} (1984). Econometric analysis of qualitative response models.
{\it Handbook of Econometrics, Vol.2}, 1395-1457, Amsterdam: Elsevier.


\bibitem{r1}
\textsc{Nolan, D. and Pollard, D.} (1987). U-processes: rates of convergence.
{\it Ann. Statist}, {\bf 15}, 780-799.

\bibitem{r1}
\textsc{Powell, J. L.} (1984). Least absolute deviations estimation
for the censored regression model. {\it J. Econometrics}, {\bf 25}, 303-325.

\bibitem{r1}
\textsc{Sherman, R. P.} (1993). The limit distribution of the maximum
rank correlation estimator. {\it Econometrica}, {\bf 61},
123-137.

\bibitem{r1}
\textsc{Sherman, R. P.} (1994a). U-processes in the analysis of a generalized semiparametric
regression estimator.
{\it Econometric Theory}, {\bf 10}, 372-395.

\bibitem{r1}
\textsc{Sherman, R. P.} (1994b). Maximum inequalities for degenerate U-processes
with applications to maximization estimators.
{\it Ann. Statist.}, {\bf 22}, 439-459.

\bibitem{r1}
\textsc{Tobin, J.} (1958). Estimation of relationships for limited dependent variables.
{\it Econometrica}, {\bf 26}, 24-26.

\bibitem{r1}
\textsc{van der Vaart, A.W.} (1998). {\it Asymptotic Statistics.},
New York: Cambridge University Press.

\end{thebibliography}
\end{document}